\documentclass{lipics-v2021}

\usepackage{amsmath}
\usepackage{amssymb}
\usepackage{xspace}
\usepackage{etoolbox}
\usepackage{xstring}
\usepackage{tikz}
\usetikzlibrary{calc}
\usetikzlibrary{intersections}
\gdef\fignphardnessscale{0.809}
\nolinenumbers

\title{The Complexity of Drawing Graphs on~Few~Lines~and~Few~Planes} %

\author{Steven Chaplick}{Maastricht University, The Netherlands}{s.chaplick@maastrichtuniversity.nl}{0000-0003-3501-4608}{Partially supported by DFG project Wo 758/11-1.}

\author{Krzysztof Fleszar}{Institute of Informatics, University of
  Warsaw,
  Poland}{kfleszar@mimuw.edu.pl}{0000-0002-1129-3289}{Supported by
  Conicyt PCI PII 20150140 and Millennium Nucleus Information and
  Coordination in Networks RC130003.}

\author{Fabian Lipp}{Institut f\"ur Informatik, Universit\"at
  W\"urzburg,
  Germany}{fabian.lipp@uni-wuerzburg.de}{0000-0001-7833-0454}{Supported
  by a Cusanuswerk PhD scholarship.}

\author{Alexander Ravsky}{Pidstryhach Institute for Applied Problems of
  Mechanics and Mathematics, National~Academy~of Sciences of Ukraine,
  Lviv, Ukraine}{alexander.ravsky@uni-wuerzburg.de}{}{Supported by DFG
  project Wo 758/9-1.}

\author{Oleg~Verbitsky}{Institut f\"ur Informatik, Humboldt
  Universit\"at,
  Germany}{verbitsk@informatik.hu-berlin.de}{0000-0002-9524-1901}{Supported
  by DFG grant VE 652/1-2.  On leave from the IAPMM, Lviv, Ukraine.}

\author{Alexander~Wolff}{Institut f\"ur Informatik, Universit\"at
  W\"urzburg, Germany}{}{0000-0001-5872-718X}{}

\authorrunning{S.~Chaplick, K.~Fleszar, F.~Lipp, A.~Ravsky,
  O.~Verbitsky, and A.~Wolff}

\relatedversion{A preliminary version of this paper appeared in Proc.\
  WADS 2017~\cite{ChaplickFLRVW16b}.}

\hideLIPIcs



\graphicspath{{figures/}}

\newcounter{claim}
\newenvironment{claim}{\refstepcounter{claim}%
\par\medskip\par\noindent{\textit{Claim~\theclaim.~}}}%
{\par\smallskip\par}
\newenvironment{subproof}{\par\noindent{\textit{Proof of Claim~\theclaim.~}}}%
{$\,\triangleleft$\par\medskip\par}

\newcommand{\function}[2]{\colon #1 \rightarrow #2}

\newcommand{\reals}{\ensuremath{\mathbb{R}}\xspace}

\newcommand{\cP}{\ensuremath{\mathrm{P}}\xspace}
\newcommand{\cNP}{\ensuremath{\mathrm{NP}}\xspace}
\newcommand{\cPSPACE}{\ensuremath{\mathrm{PSPACE}}\xspace}
\newcommand{\erclass}{\ensuremath{\exists\mathbb{R}}}
\newcommand{\feq}{\stackrel{\mbox{\tiny def}}{=}}
\newcommand{\E}{\exists}

\newcommand{\PSTR}{\textsc{Pseudoline Stretchability}\xspace}

\DeclareMathOperator{\segm}{seg}
\DeclareMathOperator{\slop}{slope}
\DeclareMathOperator{\vt}{vt}

\DeclareMathOperator{\lvaname}{lva}
\newcommand{\lva}[1]{\lvaname(#1)}

\newcommand{\PT}{\ensuremath{P_\mathrm{T}}\xspace}
\newcommand{\PF}{\ensuremath{P_\mathrm{F}}\xspace}
\newcommand{\PTSAT}{\textsc{Planar 3-Sat}\xspace}
\newcommand{\PCTSAT}{\textsc{Planar Cycle 3-Sat}\xspace}
\newcommand{\PPOITSAT}{\textsc{Positive Planar 1-in-3-Sat}\xspace}
\newcommand{\PPCOITSAT}{\textsc{Positive Planar Cycle 1-in-3-Sat}\xspace}
\newcommand{\OITSAT}{\textsc{1-in-3-Sat}\xspace}
\newcommand{\TSAT}{\textsc{3-Sat}\xspace}
\newcommand{\AGR}{\textsc{Arrangement Graph Recognition}\xspace}

\newcommand{\enquote}[1]{``#1''}

\begin{document}

\maketitle

\begin{abstract}
  It is well known that any graph admits a crossing-free straight-line
  drawing in~$\reals^3$ and that any planar graph admits the same even
  in~$\reals^2$.  For a graph~$G$ and $d \in \{2,3\}$, let
  $\rho^1_d(G)$ denote the smallest number of lines in~$\reals^d$
  whose union contains a crossing-free straight-line drawing of~$G$.
  For $d=2$, the graph~$G$ must be planar.  Similarly, let
  $\rho^2_3(G)$ denote the smallest number of \emph{planes}
  in~$\reals^3$ whose union contains a crossing-free straight-line
  drawing of~$G$.

  We investigate the complexity of computing these three
  parameters and obtain the following hardness and algorithmic
  results.
  \begin{itemize}
  \item For $d\in\{2,3\}$, we prove that deciding whether
    $\rho^1_d(G)\le k$ for a given graph~$G$ and integer~$k$ is
    $\erclass$-complete.
  \item Since $\cNP\subseteq\erclass$, deciding $\rho^1_d(G)\le
    k$ is \cNP-hard for $d\in\{2,3\}$.  On the positive side, we show
    that the problem is fixed-parameter tractable with respect to~$k$.
   \item Since $\erclass\subseteq\cPSPACE$, both $\rho^1_2(G)$
     and $\rho^1_3(G)$ are computable in polynomial space. On the
     negative side, we show that %
     drawings that are optimal with respect to $\rho^1_2$ or $\rho^1_3$
     sometimes require irrational coordinates.
  \item We prove that deciding whether $\rho^2_3(G)\le k$ is
    \cNP-hard for any fixed $k \ge 2$.  Hence, the problem is not
    fixed-parameter tractable with respect to~$k$ unless
    $\cP=\cNP$.
  \end{itemize}
\end{abstract}

\section{Introduction}

As is well known, any graph can be drawn in $\reals^3$ without
crossings so that all edges are segments of straight lines. Suppose
that we are allowed to draw edges only on a limited number of lines.
How many lines do we need for a given graph?

For planar graphs, a similar question makes sense also in $\reals^2$
since planar graphs admit straight-line drawings in $\reals^2$ by the
Wagner--F\'ary--Stein theorem. Let $\rho^1_3(G)$ denote
the minimum number of lines
which can cover a drawing of~$G$ in $\reals^3$. For a planar graph $G$, we denote the corresponding
parameter in $\reals^2$ by $\rho^1_2(G)$.  For example,
Fig.~\ref{fig:T8} shows a drawing of a planar 3-tree that lies in the
union of 14 lines, whereas Fig.~\ref{fig:T8-opt} shows a drawing of the
same graph that lies in the union of only 13 lines.
In this paper,
we restrict ourselves to straight-line and crossing-free drawings of
graphs in $\mathbb{R}^d$ for $d \in \{2,3\}$.  Note that a straight-line
drawing is completely determined by the location of the vertices.
We insist that no two vertices are mapped to the same point.

The study of the parameters $\rho^1_2(G)$ and $\rho^1_3(G)$ was posed
as an open problem by Durocher et al.~\cite{durocher2013note}.
The two parameters are related to several challenging
graph-drawing problems such as small-area or small-volume drawings~\cite{DujmovicW13}, layered or track drawings~\cite{dpw-tlg-DMTCS04},
and drawing graphs with low visual complexity~\cite{s-dgfa-JGAA15}.
Recently, we studied the extremal values of $\rho^1_3$
and $\rho^1_2$ for various classes of
graphs and examined their relations to other characteristics of
graphs~\cite{ChaplickFLRVW20}.
In particular, we showed that there are planar graphs whose
$\rho^1_3$-value is much smaller than their $\rho^1_2$-value.
Determining the exact values of $\rho^1_3(G)$ and $\rho^1_2(G)$
for a particular graph~$G$ seems to be tricky even for trees.

In fact, the setting that we suggested is more
general~\cite{ChaplickFLRVW20}.  Let $1 \le l < d$.
We define the \emph{affine cover number $\rho^l_d(G)$} as the minimum
number of $l$-dimensional planes in~$\mathbb R^d$ such that~$G$ has a
straight-line drawing that is contained in the union of these
planes. We suppose that $l\le2$ as otherwise $\rho^l_d(G)=1$.

In previous work, we have shown that we can focus on $d\le3$ as every
graph~$G$ can be drawn in
3-space as efficiently as in higher dimensions, that is,
$\rho^l_d(G)=\rho^l_3(G)$ for any $d>3$~\cite{ChaplickFLRVW20}.
This implies that, besides the
\emph{line cover numbers} $\rho^1_2(G)$ in 2D and
$\rho^1_3(G)$ in 3D, the only interesting affine cover number is the
\emph{plane cover number} $\rho^2_3(G)$ of~$G$.
Note that $\rho^2_3(G)=1$ if and only if $G$ is planar.
Let $K_n$ denote the complete graph on $n$ vertices. For the smallest
non-planar graph $K_5$, we have $\rho^2_3(K_5)=3$; see
Fig.~\ref{fig:K5}.
\begin{figure}[b]
  \begin{subfigure}[b]{.3\linewidth}
    \centering
    \includegraphics[page=1]{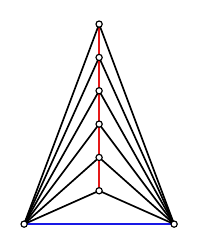}
    \caption{drawing covered by 14 lines}
    \label{fig:T8}
  \end{subfigure}
  \hfill
  \begin{subfigure}[b]{.3\linewidth}
    \centering
    \includegraphics[page=2]{3tree}
    \caption{drawing covered by 13 lines}
    \label{fig:T8-opt}
  \end{subfigure}
  \hfill
  \begin{subfigure}[b]{.32\linewidth}
    \centering
    \includegraphics[page=2]{K4+K5}
    \caption{$\rho^2_3(K_5)=3$}
    \label{fig:K5}
  \end{subfigure}
  \caption{Two drawings of the same planar 3-tree in 2D and a drawing
    of~$K_5$ in 3D.  In the 3D drawing~(c), the four white vertices
    lie in a common (gray) plane, whereas the black vertex lies above
    it.}
  \label{fig:K4+K5}
\end{figure}
The parameters $\rho^2_3(K_n)$
are not so easy to determine even for small values of $n$.
We have shown that
$\rho^2_3(K_6)=4$, $\rho^2_3(K_7)=6$, and $6\le\rho^2_3(K_8)\le 7$~\cite{ChaplickFLRVW20}.
Moreover, using Steiner systems we have shown that
$\frac{n^2-n}{12} \le \rho^2_3(K_n)\le \frac{n^2+5n+6}{6}$ for all~$n$.

The present paper focuses on the complexity of computing
affine cover numbers. A good starting point is to observe that, for given
$G$ and $k$, the statement $\rho^l_d(G)\le k$ can be expressed by
a first-order formula about the reals of the form $\exists x_1\ldots\exists x_m\Phi(x_1,\ldots,x_m)$,
where the quantifier-free subformula $\Phi$ is written using the constants $0$ and $1$,
the basic arithmetic operations,
and the order and equality relations. If, for example, $l=1$, then we just have to write
that there are $k$ pairs of points, determining a set $\mathcal{L}$ of $k$ lines,
and that there are $n$ points representing the vertices of $G$ such that
the segments corresponding to the edges of $G$ lie on the lines in $\mathcal{L}$
and do not cross each other. This observation shows that deciding whether or not
$\rho^l_d(G)\le k$ reduces in polynomial time to the decision problem (Hilbert's \emph{Entscheidungsproblem})
for \emph{the existential theory of the reals}. The problems admitting such a reduction
form the complexity class $\erclass$ introduced by Schaefer~\cite{Schaefer09},
whose importance in computational geometry has been recognized recently~\cite{Cardinal15,Matousek14,SchaeferS15}.
In the complexity-theoretic hierarchy, this class occupies a position between
\cNP and \cPSPACE.
It possesses natural complete problems like the decision version of
the rectilinear crossing number~\cite{Bienstock91} and
the recognition of segment intersection graphs~\cite{KratochvilM94},
unit disk graphs~\cite{KangM12}, and point visibility graphs~\cite{CardinalH17}.
Related questions deal with the realizability of abstract simplicial
complexes
\cite{mtw-hescrd-JEMS11,mstw-e3sd-JACM18,akm-gecerc-arXiv21},
polyhedral complexes~\cite{akklsvw-dgps-SoCG21}, graphs, and
hypergraphs \cite{erssw-rghtp-GD19} in three and higher dimensions.

Below, we summarize our results on the computational complexity of the
affine cover numbers.

\subparagraph*{The complexity of the line cover numbers in 2D and 3D.}
We begin by showing that it is $\erclass$-hard to compute, for a given graph~$G$, its line
cover numbers $\rho^1_2(G)$ and $\rho^1_3(G)$; see Section~\ref{s:rho1213}.

Our proof uses some ingredients from a paper of Durocher et
al.~\cite{durocher2013note} who showed that it is \cNP-hard to compute
the \emph{segment number} $\segm(G)$ of a graph~$G$.
This parameter was introduced by Dujmovi{\'c} et al.~\cite{desw-dpgfs-CGTA07}
as a measure of the visual complexity of a planar graph.
A \emph{segment} in a straight-line drawing of a graph~$G$
is an inclusion-maximal connected path of edges of $G$ lying on a line, and the \emph{segment number} $\segm(G)$ of a planar graph $G$ is the minimum
number of segments in a straight-line drawing of $G$ in the plane.
Note that while $\rho^1_2(G)\le\segm(G)$, the two parameters can be
far apart as demonstrated, for instance, by a graph with $m$ isolated edges.
For connected graphs, we have shown in previous work~\cite{ChaplickFLRVW20}
that $\segm(G) \in O(\rho^1_2(G)^2)$
and that this bound is optimal as
there exist planar triangulations with $\rho^1_2(G) \in O(\sqrt n)$ and
$\segm(G) \in \Omega(n)$.
Despite this difference, we follow Durocher et al.~\cite{durocher2013note} to some
extent in that we also reduce from \AGR (see Theorem~\ref{thm:rho12-hard}).

Another noteworthy related result is the $\erclass$-hardness of computing the
\emph{slope number} $\slop(G)$ of a planar graph $G$, which has recently been
established by Hoffmann~\cite{Hoffmann17}.  The value of $\slop(G)$ is
equal to the minimum possible number of slopes in a straight-line drawing of $G$.
It should be noted that, whereas $\rho^1_2(G)\ge\slop(G)$, the two parameters are generally
unrelated. For example, if $G$ is a nested-triangle graph, then
the former parameter is linear while the latter is bounded by a constant.

\subparagraph*{Parameterized complexity of computing the line cover numbers in 2D  and 3D.}
It follows from the inclusion $\cNP\subseteq\erclass$ that the
decision problems~\enquote{$\rho^1_2(G)\le k$?}
and~\enquote{$\rho^1_3(G)\le k$?} are \cNP-hard
if $k$ is given as a part of the input.
On the positive side, in Section~\ref{s:fpt}, we show that both
problems are fixed-parameter tractable (FPT).
To this end, we first describe a linear-time kernelization procedure
that reduces the given graph to one of size $O(k^4)$. Then, in $k^{O(k^2)}$ time,
we carefully solve the problem on this reduced
instance by using the exponential-time decision procedure for the existential theory of the
reals by Renegar~\cite{Renegar92a,Renegar92b,Renegar92c} (see
Theorem~\ref{thm:renegar} in Section~\ref{s:rho1213}) as a
subroutine.
To the best of our knowledge, this is the first application of
Renegar's algorithm
for obtaining an FPT result, in particular, in the area of graph drawing where
FPT algorithms are widely known.
For comparison note that, for every fixed $k$, the decision
problem~\enquote{$\slop(G)\le k$?} is in \cNP~\cite{Hoffmann17}.

\subparagraph*{Realizability of $\rho^1_d$-optimal drawings.}
Since $\erclass$ belongs to $\cPSPACE$ (as shown by
Canny~\cite{Canny88}), the parameters $\rho^1_d(G)$ for both $d=2$ and
$3$ are computable in polynomial space. We show,
however, that \emph{constructing} a $\rho^1_2$-optimal drawing of a given planar graph $G$
can be an unfeasible task by the following reason:
there is a planar graph $G$ such that every $\rho^1_2$-optimal drawing
of $G$ requires irrational coordinates (see Theorem~\ref{thm:irrational-coord-graph}).

This result shows that, even if a graph $G$ is known to be drawable on $k$ lines,
it may happen that $G$ does not admit a $k$-line drawing on the integer grid.
Nevertheless, in FPT time we are always able to produce a \emph{combinatorial description}
of an optimal drawing (see Theorem~\ref{thm:fpt-main^2}).

\subparagraph*{The complexity of the plane cover number.}
Though the decision problem~\enquote{$\rho^2_3(G) \le k$?}
also belongs to $\erclass$, its complexity status is different from
that of the line cover numbers.  In
Section~\ref{sec:complexity-rho23}, we establish the \cNP-hardness of
deciding whether $\rho^2_3(G)\le k$ for any fixed $k\ge2$, which
excludes an FPT algorithm for this problem unless $\cP=\cNP$.
To show this, we first prove \cNP-hardness of \PPCOITSAT (a new
problem of \PTSAT type), which we think is of independent
interest.

In subsequent work, Pilz~\cite{p-p3sat-DMTCS19} showed that another variant
of this type is NP-hard, namely the restriction of \PTSAT to instances
where a Hamiltonian cycle can be added to the bipartite
variable--clause graph without violating planarity.  While in his
case, the cycle goes through all clause vertices and then through all
variable vertices, our restriction only asks for a cycle through all
clause vertices.  On the other hand, Pilz showed that his variant
remains NP-hard for the \textsc{Monotone} variant (where no clause may
combine negative and positive literals), while we show this for the
\textsc{Positive 1-in-3} variant (where all literals are positive and
the task is to decide whether there exists an assignment of truth
values to the variables such that exactly one variable in each clause
is true).

\subparagraph*{Weak affine cover numbers.}
We previously defined the \emph{weak affine cover number}
$\pi^l_d(G)$ of a graph~$G$ similarly to $\rho^l_d(G)$ but under the weaker
requirement that the $l$-dimensional planes in $\reals^d$ whose number
has to be minimized contain the \emph{vertices} (and not
necessarily the edges) of~$G$~\cite{ChaplickFLRVW20}.
Based on our combinatorial characterization of~$\pi^1_3$
and~$\pi^2_3$~\cite{ChaplickFLRVW20}, we show in
Section~\ref{s:weak-appendix} that the decision
problem~\enquote{$\pi^l_3(G)\le2$?} is \cNP-complete for $l \in \{1,2\}$.
Recently, Erhardt~\cite{e-wlcn3d-MTh21} strengthened our result for
$l=1$.  He showed that deciding whether $\pi^1_3(G)\le2$ is NP-hard
even if $G$ has maximum degree~$5$ or if $G$ is planar and has maximum
degree~$6$.  For general graphs, his result is tight because
Matsumoto~\cite{m-bvla-JGT90} has proven that for every
graph~$G \ne K_5$ of maximum degree~$4$, it holds that
$\pi^1_3(G)\le2$.  Erhardt also
showed that, for any constant~$c$, the problem of deciding whether
$\pi^1_3(G)\le c$ is FPT parameterized by the treewidth of~$G$.

Further, we prove that it is \cNP-hard to approximate
$\pi^l_3(G)$ within a factor of $O(n^{1-\epsilon})$, for any
$\epsilon>0$.
Recently, Biedl et al.~\cite{bfw-lpcnr-GD19} showed that it is
\cNP-complete to decide whether $\pi^1_2(G)=2$ for a given planar
graph~$G$.  Firman et al.~\cite{flsw-ewlc2-GD18} gave formulations of
the problem as an integer linear program and as an instance of Boolean
satisfiability and compared them on test data; they also examined the
drawability of several special graphs.

Given a graph~$G$, we write $V(G)$ for its vertex set and $E(G)$ for
its edge set.  For any dimension~$d>1$ and two points $x$ and $y$
in~$\reals^d$, we write $[x,y]$ for the line segment that connects~$x$
and~$y$, and we write $xy$ for the line that goes through~$x$ and~$y$.

\section{Computational Hardness of the Line Cover Numbers}
\label{s:rho1213}

In this section, we show that, for a given graph~$G$ and
an integer~$k$, the two decision problems \enquote{$\rho^1_2(G) \le k$?}
and~\enquote{$\rho^1_3(G) \le k$?} are~$\erclass$-complete.
The $\erclass$-hardness results are often established by a reduction from
the \PSTR problem: Given an arrangement
of pseudolines in the projective plane, decide whether it is \emph{stretchable}, that is,
equivalent to an arrangement of lines~\cite{Mnev85,Mnev88}.
Our reduction is based on an argument of Durocher et
al.~\cite{durocher2013note} who designed a reduction of
the \AGR problem, defined below, to the problem of
computing the segment number of a graph.

A \emph{simple line arrangement} is a set $\mathcal{L}$ of $k$ lines in $\mathbb R^2$
such that each pair of lines crosses each other in a point, their
\emph{crossing point}, and no three lines share a common point.
In the following, we assume that every line arrangement is simple.
We define the \emph{arrangement graph} for a (simple) line arrangement
as follows~\cite{BoseEW03} (see Fig.~\ref{fig:arrangement-graph-orig}):
The vertices correspond to the crossing points of the lines and two vertices
are adjacent in the graph if and only if they are adjacent along some line.
Note that the arrangement graph for a simple line arrangement consisting of~$k$
lines has $\binom{k}{2}=k(k-1)/2$ vertices and $k(k-2)$ edges
(the latter because each of the $k$ lines hosts $k-2$ edges between
$k-1$ crossing points).
The \AGR problem is to decide whether a given graph is the arrangement
graph of some line arrangement.

\begin{figure}[b]
    \newcommand{\figArrangementPaths}{
        \path[name path=l1] (-0.1,0) -- (5.25,0);
        \path[name path=l2] (0,-0.5) -- (4,2.3);
        \path[name path=l3] (1.65,-0.85) -- (2.45,2.75);
        \path[name path=l4] (1.3,2.75) -- (5.15,-0.65);

        \path[name intersections={of={l1 and l2},by=v1}];
        \path[name intersections={of={l1 and l3},by=v2}];
        \path[name intersections={of={l1 and l4},by=v3}];
        \path[name intersections={of={l2 and l3},by=v4}];
        \path[name intersections={of={l2 and l4},by=v5}];
        \path[name intersections={of={l3 and l4},by=v6}];
    }
    \newcommand{\figArrangementLines}{
        \draw (-0.1,0) -- (5.25,0);
        \draw (0,-0.5) -- (4,2.3);
        \draw (1.65,-0.85) -- (2.45,2.75);
        \draw (1.3,2.75) -- (5.15,-0.65);
    }
    \begin{subfigure}[t]{.45\textwidth}
        \centering
        \begin{tikzpicture}[scale=.9]
            \figArrangementPaths
            \begin{scope}[color=green!70!black,opacity=1,ultra thick]
                \foreach \x in {v1,v2,v3,v4,v5,v6}
                    \fill[draw] (\x) circle(2pt);
                \draw (v1) -- (v2);
                \draw (v1) -- (v4);
                \draw (v2) -- (v3);
                \draw (v2) -- (v4);
                \draw (v3) -- (v5);
                \draw (v5) -- (v6);
                \draw (v4) -- (v5);
                \draw (v4) -- (v6);
            \end{scope}
            \figArrangementLines
        \end{tikzpicture}
        \caption{A line arrangement (black lines) with the
          corresponding arrangement graph $G$ (in green).}
        \label{fig:arrangement-graph-orig}
    \end{subfigure}
    \hfill
    \begin{subfigure}[t]{.5\textwidth}
        \centering
        \begin{tikzpicture}[scale=.9]
            \figArrangementPaths
            \begin{scope}[color=blue!60!white,opacity=1,ultra thick]
            \foreach \x in {v1,v2,v3,v4,v5,v6}
                \fill[draw] (\x) circle(2pt);
            \draw (v1) -- (v2);
            \draw (v1) -- (v4);
            \draw (v2) -- (v3);
            \draw (v2) -- (v4);
            \draw (v3) -- (v5);
            \draw (v5) -- (v6);
            \draw (v4) -- (v5);
            \draw (v4) -- (v6);
            \foreach \x/\y in {v1/v2,v1/v4,v2/v4,v3/v5,
                               v3/v2,v5/v4,v6/v4,v6/v5}
                \fill[draw] (\x) -- ($(\x)!-0.5cm!(\y)$) circle(2pt);
            \end{scope}
            \figArrangementLines
        \end{tikzpicture}
        \caption{The modified arrangement graph $G'$ (in blue) is
          obtained from~$G$ by adding two leaves on each line.}
        \label{fig:arrangement-graph-modified}
    \end{subfigure}

    \caption{A simple line arrangement of four lines (depicted in black) with
        the corresponding arrangement graph and the modification used in the
        proof of Theorem~\ref{thm:rho12-hard}.}
    \label{fig:arrangement-graph}
\end{figure}

Bose et al.~\cite{BoseEW03} showed that this problem is \cNP-hard
by reduction from a version of \PSTR for the Euclidean plane,
whose \cNP-hardness was proved by Shor~\cite{Shor91}.
It turns out that \AGR is actually an $\erclass$-complete problem~\cite[page 212]{Eppstein14}.
This stronger statement follows from the fact
that the Euclidean \PSTR is $\erclass$-hard as well as the original projective version~\cite{Matousek14,Schaefer09}.

\begin{theorem}\label{thm:rho12-hard}
  Given any planar graph~$G$ and any integer~$k$,
  the decision problems~\enquote{$\rho^1_2(G)\le k$?}
  and~\enquote{$\rho^1_3(G)\le k$?} are~$\erclass$-hard.
\end{theorem}

\begin{proof}
We first treat the 2D case.
We show hardness by reduction from \AGR.
Let~$G$ be an instance of this problem consisting of more than one vertex; otherwise we accept.
Here and below, by \emph{acceptance} (resp.\ \emph{rejection}) we mean mapping an instance of \AGR
to a fixed yes-instance (resp.\ no-instance) of the decision problem
\enquote{$\rho^1_2(G)\le k$?}.
If~$G$ is an arrangement graph, then there must be an integer $\ell \ge 3$ such that $G$ consists of
$\ell(\ell-1)/2$ vertices and $\ell(\ell-2)$ edges, and each of its vertices has degree $2$, $3$, or $4$.
So, we first check these simple conditions to determine $\ell$ and reject $G$ if one of them fails.
Let $G'$ be the graph obtained from~$G$ by adding one tail
(that is, an edge to a new degree-$1$ vertex) to
each degree-$3$ vertex and two tails to each degree-$2$ vertex, see
Fig.~\ref{fig:arrangement-graph-modified}.
Thus, every vertex of~$G'$ has degree~$1$ or~$4$.
Note that, if $G$ is an arrangement graph, then there are exactly $2\ell$ tails in~$G'$
(two for each line) -- if this is not true we can already safely reject $G$.
We now pick $k = \ell$, and show that~$G$ is an arrangement graph if and only if $\rho^1_2(G')
\leq k$.

For one %
direction, assume that $G$ is an arrangement graph.  By our
choice of~$k$, it is clear that~$G$ corresponds to a line arrangement
of~$k$ lines.
Note that all edges of $G$ lie on these $k$ lines and the tails of $G'$ can be
added without increasing the number of lines.
Hence, $\rho^1_2(G') \leq k$.

For the other direction, assume that $\rho^1_2(G') \leq k$, and let
$\gamma'\function{V(G')}{\reals^2}$ be a crossing-free straight-line
drawing of~$G'$ 
on a set $\mathcal{L}$ of $\rho^1_2(G')$ lines.  Let $\gamma$ be the
restriction of $\gamma'$ to $V(G)$.  The graph~$G'$ contains
$\binom{k}{2}$ degree-$4$ vertices (which are exactly the vertices in
$V(G)$).  As each of these vertices lies on the crossing point of two
lines in~$\mathcal{L}$, we need $k$ lines to get enough crossing
points, that is, $|\mathcal{L}| = k$, every two lines in $\mathcal{L}$
cross each other, and no three lines in $\mathcal{L}$ cross in the
same point.  Thus, $\mathcal{L}$ is a simple line arrangement. Let
$A_\mathcal{L}$ be the geometric realization of the arrangement graph
of $\mathcal{L}$, that is, $V(A_\mathcal{L})$ consists of the
$\binom{k}{2}$ crossing points of lines in $\mathcal{L}$. The mapping
$\gamma$ is a bijection from $V(G)$ to $V(A_\mathcal{L})$, and it
suffices to show that $\gamma$ is an isomorphism from~$G$ to~$A_\mathcal{L}$.
Indeed, let~$u$ and $v$ be vertices that are adjacent in~$G$.
Since~$\gamma'$ is a straight-line drawing of~$G'$ on~$\mathcal{L}$,
the line segment $[\gamma(u),\gamma(v)]$ is a segment between
two crossing points on a line in $\mathcal{L}$.  This segment does not
contain any other crossing point because any such point would be the
image $\gamma(w)$ of some vertex $w\not\in\{u,v\}$ of~$G$ (which would
contradict $u$ and $v$ being adjacent in~$G$).  It follows that
$\gamma(u)$ and $\gamma(v)$ are adjacent in $A_\mathcal{L}$.  Since
$G$ and $A_\mathcal{L}$ have the same number $k(k-2)$ of edges,
$\gamma$ is an isomorphism.

Now we turn to 3D.
Let $G$ be a planar graph (that passed all our initial rejection tests) and let $G'$ be the corresponding graph augmented as above.
We show that~$\rho^1_3(G')\le k$ if and only if~$\rho^1_2(G') \le k$, which yields that
deciding~$\rho^1_3(G')$ is also $\erclass$-hard.
For the first direction, if~$G'$ cannot be drawn on~$k$ lines in 3-space, then clearly it neither can in 2-space.
For the second direction, assume that~$G'$ can be drawn on~$k$ or fewer lines in 3-space.
Since $G'$ has $\binom{k}{2}$ vertices of degree~$4$, each of them
must be a crossing point of two lines. It follows that we have exactly~$k$ lines and that each of the $k$ lines
crosses all the others. Fix any two of the lines and consider the plane
that they determine.  Then all $k$ lines must lie in this plane, which
shows that also~$\rho^1_2(G') \leq k$.
\end{proof}

It remains to notice that the decision problems under consideration
lie in the complexity class~\erclass. To this end,
we transform the inequalities $\rho^l_d(G)\le k$ into first-order existential expressions about the reals.
More precisely, we call a first-order formula \emph{existential} if it is written in prenex normal form,
that is, as a prefix of quantifiers applied to a quantifier-free part,
and the prefix consists only of existential quantifiers.
Though this transformation is direct and elementary, we give some details in the proof of the following
lemma, as they are relevant also to the proof of Theorem~\ref{thm:fpt-main^2} in Section~\ref{s:fpt}.

\begin{lemma}
  \label{lem:inER}
    Given any graph~$G$ with at least one edge and any positive
    integer~$k$, the following decision problems belong to the
    complexity class~$\erclass$:
    \begin{enumerate}
    \item \label{enum:rho12-ER}
    \enquote{$\rho^1_2(G)\le k$?} (assuming that~$G$ is planar);
    \item
    \enquote{$\rho^1_3(G)\le k$?};
    \item
    \enquote{$\rho^2_3(G)\le k$?}.
    \end{enumerate}
\end{lemma}

\begin{proof}
  We have to express the conditions $\rho^l_d(G) \leq k$ in the
  formalism of the existential theory of the reals.  This can be done
  rather straightforwardly applying elementary tools of Cartesian
  geometry.  We provide the necessary details for the parameter
  $\rho^1_2(G)$ and sketch how these arguments need to be modified for
  the other two parameters.

Given three points $a=(x_1,y_1)$, $b=(x_2,y_2)$, $c=(x_3,y_3)$ in the plane, let
$$
\chi(a,b,c)=\left|
  \begin{array}{ccc}
    x_1&y_1&1\\
    x_2&y_2&1\\
    x_3&y_3&1
  \end{array}
\right|
$$
be the scalar triple product of three 3-dimensional vectors
$(x_1,y_1,1)$, $(x_2,y_2,1)$, and $(x_3,y_3,1)$. As is well known,
the following conditions are equivalent:
\begin{itemize}
\item
$\chi(a,b,c)>0$;
\item
the sequence of vectors $(x_1,y_1,1)$, $(x_2,y_2,1)$, and $(x_3,y_3,1)$
forms a right-handed system in~$\reals^3$;
\item
$a\ne b$, and the point $c$ lies in the left half-plane with respect to the oriented line $\vec{ab}$;
\item
the points $a$, $b$, and $c$ are pairwise distinct, non-collinear, and occur counterclockwise in the circumcircle
of the triangle~$\triangle abc$ in this order.
\end{itemize}
Moreover, $\chi(a,b,c)=0$ if and only if the points $a$, $b$, and $c$
are collinear, including the case that some of them coincide.
Furthermore, the point $a$ lies on the line segment~$[b,c]$ if and
only if
the following relation is fulfilled:
\begin{eqnarray*}
B(a,b,c) \;\feq\; \chi(a,b,c)=0&\wedge&
\left((x_2-x_1)^2+(y_2-y_1)^2\le(x_3-x_2)^2+(y_3-y_2)^2\right)\\
&\wedge& \left((x_3-x_1)^2+(y_3-y_1)^2\le(x_3-x_2)^2+(y_3-y_2)^2\right),
\end{eqnarray*}
that is, the points $a$, $b$, and $c$ are collinear and the sum of the
distances from~$a$ to~$b$ and from~$a$ to~$c$ does not exceed the
distance of~$b$ and~$c$.

Two line segments $[a,b]$ and $[c,d]$ do not intersect
if and only if
the points $a$, $b$, $c$, and~$d$ satisfy the following relation:
\begin{align*}
  D(a,b,c,d) \;\feq\;& \;\big(\chi(a,b,c)\chi(a,b,d)>0\big) \;\vee\;
  \big(\chi(c,d,a)\chi(c,d,b)>0\big) \;\vee{}\\
  & \;\big(\chi(a,b,c)=\chi(a,b,d)=0\wedge \neg B(a,c,d) \wedge
  \neg B(b,c,d)\wedge \neg B(c,a,b)\wedge \neg B(d,a,b)\big),
\end{align*}
that is, $c$ and $d$ are on the same side of the line~$ab$,
or $a$ and $b$ are on the same side of the line~$cd$,
or $[a,b]$ and $[c,d]$ are disjoint segments lying on the same line.

We assume that the given graph~$G$ has the vertex set $\{1,\dots,n\}$, and
we let~$E$ denote the edge set of~$G$.
We have to express the fact that there are $n$ pairwise distinct points $v_1,\ldots,v_n$
lying on $k$ lines $\ell_1,\ldots,\ell_k$ that determine a straight-line drawing of the graph $G$.
Each $\ell_i$ can be represented by a pair of points $p_i$ and $q_i$ lying on this line.
Our existential statement about the reals begins, therefore, with the quantifier prefix
$\E v_1\ldots\E v_n\,\E p_1\E q_1\ldots \E p_k\E q_k$, where quantification $\E a$ over a point $a=(x,y)$
means the quantifier block $\E x\E y$.
Given this prefix of quantifiers, we use the following subformula to
express that every edge of $G$ lies on one of the $k$ lines:
$$
\bigwedge_{i\ne j}v_i\ne v_j \quad\wedge\quad \bigwedge_{l=1}^kp_l\ne q_l
\quad\wedge\quad \bigwedge_{\{i,j\}\in E} \;\; \bigvee_{l=1}^k
\big(B(v_i,p_l,q_l)\wedge B(v_j,p_l,q_l)\big),
$$
where $a\ne b$ for points $a=(x_1,y_1)$ and $b=(x_2,y_2)$ is an
abbreviation for $x_1\ne y_1\vee x_2\ne y_2$.

It remains to ensure that there are no
edge crossings.  To this end, we simply write
\[
  \bigwedge_{\{i,j\},\{i',j'\}\in E,\,\{i,j\}\cap\{i',j'\}=\emptyset}
  D(v_i,v_j,v_{i'},v_{j'}) \quad \wedge
  \bigwedge_{\{i,j\},\{i,j'\}\in E,\,j\ne j'} \big(\neg
B(v_{j'},v_i,v_j)\wedge\neg B(v_j,v_i,v_{j'})\big).
\]

This finishes our description of an existential first-order sentence
about the reals that encodes the condition $\rho^1_2(G) \le k$.  In this
sentence, each polynomial has total degree at most~$4$ and each
variable has an integer coefficient whose absolute value is bounded
by~$2$.

In order to encode ``$\rho^1_3(G)\le k$'' or ``$\rho^2_3(G)\le k$'',
the polynomials will have slightly larger total degree and the
variables will have coefficients with slightly larger absolute values.
An obvious difference to the case of $\rho^1_2(G)$ that we treated
above is that now we need three real coordinates to represent a point
and four coefficients $u_1,\dots,u_4$ to specify a plane by its
canonical equation $u_1x_1+u_2x_3+u_3x_3=u_4$.
The condition that a point $(x_1,x_2,x_3)$ lies on
this plane can be expressed by plugging its coordinates into the
plane's equation.

Finally, the condition that two line segments in 3D space do not intersect
can be expressed either by using the fact that the orientation of a triple of vectors
with a common origin is determined by the determinant of the matrix of their coordinates
(allowing a simple reduction to the 2D case)
or by using an explicit arithmetic formula for the crossing point of
two lines; see, for example,~\cite[Section~V.3]{graphics-gems-book}.
\end{proof}

The proof of Lemma~\ref{lem:inER} implies
the following.  Recall that a \emph{first-order sentence} is a formula
in first-order logic where no free variables occur.

\begin{corollary}\label{cor:PhiGk}
  Let $(l,d) \in \{(1,2),(1,3),(2,3)\}$.  Given any graph $G$ with
  $n$ vertices and $m\ge1$ edges, and any positive integer $k$,
  we can construct, in time $O(n^2+m^2)$,
  an existential first-order sentence~$\Phi_{G,k}$ about the reals such that
  \begin{itemize}
  \item
    $\Phi_{G,k}$ involves $O(n^2+m^2)$ polynomials in at most
    $d \cdot (n+(l+1) \cdot k)$
    variables, each of constant total degree and with integer
    coefficients of constant absolute value, and
  \item
    $\Phi_{G,k}$ is true if and only if $\rho^l_d(G)\le k$.
  \end{itemize}
\end{corollary}

Note that the constant hidden in the big-$O$ notation does not depend
on~$k$ because, if $k \ge m$, we can directly set $\Phi_{G,k}$ to true
since in this case we can easily cover $m$ edges with $k$ lines or
planes.  The same applies to the time bound in
Theorem~\ref{thm:line-cover-kernel} below.

\section{Fixed-Parameter Tractability of the Line Cover Numbers}
\label{s:fpt}

In this section, we show that, for a graph $G$ with $m\ge1$ edges
and a positive integer~$k$,
the decision problems~\enquote{$\rho^1_2(G) \leq k$?}
and~\enquote{$\rho^1_3(G) \leq k$?}
can both be solved in FPT time parameterized by~$k$.
(The cases $k=0$ and $m=0$ are not interesting.)
Given a graph~$G$ and a drawing~$\Gamma$ of~$G$, we
define~$\mathcal{L}_\Gamma$ to be the arrangement of the supporting
lines of the line segments in~$\Gamma$.
In contrast to Section~\ref{s:rho1213},
we now consider line arrangements in~$\mathbb R^2$ and in~$\mathbb R^3$
that are not necessarily simple, that is,
crossings of more than two lines in a point are allowed.
If $|\mathcal L_\Gamma| = k$, we call $\Gamma$ a \emph{$k$-line drawing}.
We also call $\mathcal L_\Gamma$ a line cover of $G$;
for this to be true, we move any isolated vertices onto distinct points of
some line in $\mathcal L_\Gamma$, avoiding other vertices and edges of~$\Gamma$.

Our FPT algorithm follows from a simple \emph{kernelization/pre-processing}
procedure in which we reduce a given instance $(G,k)$ to a reduced instance
$(H,k)$ where $H$ has $O(k^4)$ vertices and edges,
and $G$ has a $k$-line drawing if and only if $H$ does as well.
After this reduction,
we can apply any decision procedure for the existential theory of the reals
since we have shown in Lemma~\ref{lem:inER} that the two $k$-line drawing problems
are indeed members of this complexity class.
Our kernelization approach is
given as Theorem~\ref{thm:line-cover-kernel} and our FPT result follows as
described in Corollary~\ref{cor:line-cover-fpt}.

\begin{theorem}\label{thm:line-cover-kernel}
  For~$d \in \{2,3\}$, given any graph~$G$ with~$n$ vertices
  and~$m\ge1$ edges, and any positive integer~$k$,
  the decision problem~\enquote{$\rho^1_d(G) \leq k$?}
  admits a kernel of size $O(k^4)$. Specifically, in~$O(n+m)$ time
  one can compute a graph $H$ with at most
  $1.1k^4$ vertices and at most $2k^4$ edges such that
  $\rho^1_d(G) \leq k$ if and only if $\rho^1_d(H) \leq k$.
\end{theorem}

\begin{proof}
For our kernel of~$G$, we construct a graph~$H$ with~$O(k^4)$ many vertices and
edges such that~$\rho^1_d(H) \leq k$ if and only if~$\rho^1_d(G) \leq k$.
Our idea is based on the observation that, in any drawing~$\Gamma$
of~$G$, the total number of bends
in ``long'' paths of degree-$2$ vertices is naturally
bounded by the number of crossing points of the underlying line
cover~$\mathcal{L}_\Gamma$.
This fact allows us to sufficiently shrink long subpaths of
degree-$2$ vertices without changing the drawability of the graph.

We compute our kernel in two steps.
First, we remove from~$G$ all path components.
In each of the remaining connected components, we shrink every
path consisting only of vertices of degree at most~$2$ down to
length~$\binom{k}{2}$ by performing sufficiently many edge
contractions. We can do this modification in time~$O(n+m)$ by using,
e.g., breadth-first search to first measure the length of such paths
and then to shrink them.  Let~$G'$ be the resulting graph.
Second, if $G'$ has at most $1.1k^4$ vertices and at most $2k^4$ edges,
we set~$H=G'$; otherwise we set~$H=K_{1,2k+1}$ (which we use as a fixed no-instance of small size).
Thus, the size of~$H$ is as required.

To see that~$H$ is a kernel, first observe that if~$G$ is drawable
on~$k$ lines, then~$G'$ is drawable in the same way since every
shrunken path has enough vertices to map to positions where the
drawing of the original path makes bends; note that bends can happen
only at crossing points and that their number is bounded
by~$\binom{k}{2}$.

On the other hand, if~$G'$ is drawable on~$k$ lines, then also~$G$ has
a~$k$-line drawing that we can obtain by appropriately subdividing
shrunken paths in the $k$-line drawing of~$G'$ and by placing any path
components on a single line such that they do not interfere with the
drawings of the other components.

Moreover, if~$G'$ is drawable on~$k$ lines, then also~$H$ is~-- as
then the size of~$G'$ is less than~$4k^4$, which implies~$H=G'$.
To see the bound on the size, first observe that the number of
vertices of degree larger than~$2$ is naturally bounded
by~$\binom{k}{2}$, our upper bound on the number of crossing points.
All remaining vertices, whose degree is at most~$2$, form edge-disjoint
cycles and (maximal) paths.  By our construction, the length of each
such cycle or path is bounded by~$\binom{k}{2}$.  Each of the cycles
visits at least three crossing points; thus, the number of cycles
is at most~$\binom{k}{2}/3$.  Each of the paths is adjacent to at
least one vertex of degree larger than~$2$ (otherwise it is a path
component, but we have already removed them); thus, the number of
paths is bounded by the total degree of the higher-degree vertices.
Since the higher-degree vertices have to lie on crossing points and
each line contains at most~$k-1$ higher-degree vertices,
their total degree is at most~$k \cdot 2(k-1)$.
In summary, the total number of vertices of~$G'$ is at
most~$(1+\binom{k}{2}/3+2k(k-1))\cdot\binom{k}{2}<1.1k^4$.  The total
number of edges of~$G'$ is the total vertex degree of~$G'$ over~$2$,
that is, at most~$(2k(k-1) + 2\cdot1.1 k^4)/2 < 2k^4$.

Finally, if~$H$ is drawable on~$k$ lines, then~$H=G'$
as~$K_{1,2k+1}$ cannot be drawn on~$k$ lines. Thus,~$G'$ and,
consequently,~$G$ are also drawable on~$k$ lines.  This yields, as
required, that $\rho^1_d(H) \leq k$ if and only
if~$\rho^1_d(G) \leq k$.  In other words,~$H$ is indeed a kernel
of~$G$.
\end{proof}

Note that, for computing the plane cover number~$\rho^2_3$, we cannot
get a polynomial kernel (assuming $\cP \ne \cNP$) as we show that it
is NP-hard to decide, for a given graph~$G$, whether $\rho^2_3(G)=2$
(see Theorem~\ref{thm:rhoeq2}).

Let~$G$ be a graph of~$n$ vertices and~$m$ edges, and let~$H$ be a kernel of size~$O(k^4)$ as obtained by Theorem~\ref{thm:line-cover-kernel} from~$G$ in time~$O(n+m)$.
By Corollary~\ref{cor:PhiGk}, %
the statement $\rho^1_d(H)\le k$ can be expressed
by a statement~$\Phi$ in the existential first-order
theory of the reals using $O(k^4)$ first-order variables and
$k^{O(1)}$ polynomial inequalities, each of constant total degree
and with coefficients of constant absolute value.
Furthermore, such a statement~$\Phi$ can be computed in time~$k^{O(1)}$.
If we directly apply the following decision procedure of Renegar
to~$\Phi$, we obtain an algorithm for deciding
for the original graph~$G$ whether $\rho^1_d(G)\le k$.

\begin{theorem}[Renegar~\cite{Renegar92a,Renegar92b,Renegar92c}]
  \label{thm:renegar}
  Given any existential sentence $\Phi$ about the reals, one can
  decide whether $\Phi$ is true or false in time
  \[
    (L\log L\log\log L) \cdot (PD)^{O(N)},
  \]
  where~$N$ is the number of variables occurring in~$\Phi$,~$P$ is the
  number of polynomials involved in~$\Phi$, $D$ is the maximum total
  degree over the polynomials in~$\Phi$, and~$L$ is the maximum length
  of the binary representation over the coefficients of the polynomials
  in~$\Phi$.
\end{theorem}
Using Renegar's result, our decision algorithm runs in
$k^{O(k^4)} + O(n+m)$ time. This already shows that deciding
$\rho^1_d(G)\le k$ is in FPT.  We now reduce the exponent in the
running time to $O(k^2)$.

\begin{corollary}\label{cor:line-cover-fpt}
For~$d \in \{2,3\}$, given any graph $G$ with~$n$ vertices and~$m\ge1$
edges, and any positive integer~$k$,
the decision problem~\enquote{$\rho^1_d(G) \leq k$?} can be solved
in~$k^{O(k^2)} + O(n+m)$ time, that is, in FPT time parameterized by~$k$.
\end{corollary}
\begin{proof}
First, we apply to the given graph~$G$ the kernelization procedure
from Theorem~\ref{thm:line-cover-kernel} to obtain a reduced graph $H$.
By Theorem~\ref{thm:line-cover-kernel}, $H$
has at most $1.1k^4$ vertices. %
At most $\binom{k}{2}$
of these vertices are mapped to crossing points of the
lines.  Hence we enumerate all possible $\binom{1.1k^4}{\binom{k}{2}}
\in k^{O(k^2)}$ subsets that can be mapped to crossing points.
Such a subset must contain all vertices of degree more than~2.
The paths consisting of degree-2 vertices not in the subset
are contracted. For each of these subsets, we
test whether this further reduced instance
has a $k$-line drawing using Renegar's decision algorithm.
Since this reduces the number of vertices and, hence, the number of first-order variables
to $O(k^2)$, Theorem~\ref{thm:renegar} ensures a total running time of
$k^{O(k^2)} + O(n+m)$ as required.
\end{proof}

Corollary~\ref{cor:line-cover-fpt} says that we can decide in FPT time whether or not
a graph $G$ has a $k$-line drawing.
However, this does not mean that, if such a drawing exists,
we can construct it efficiently as well.
In fact, there is a graph such that every $\rho^1_2$-optimal
drawing of this graph has a vertex with an irrational coordinate; see
Section~\ref{sec:irrational-coordinates2}. Nevertheless, we are able to
obtain a \emph{combinatorial description} of a $k$-line drawing in FPT time
whenever one exists.
We now define this concept.

First, we need an additional variant of an arrangement graph.
Let~$\mathcal{L}$ be an arrangement of $k$ distinct lines
in~$\reals^d$.  Without loss of generality, we can assume that there
is no isolated line, that is, every line in~$\mathcal{L}$ is crossed
by at least one other line in~$\mathcal{L}$.
Note that we do \emph{not} assume that every point is contained in at
most two lines.  The \emph{augmented arrangement graph}~$A_\mathcal{L}$
has crossing points of lines in~$\mathcal{L}$
as vertices, and two such points are adjacent in $A_\mathcal{L}$ if they are neighboring on a
line in $\mathcal{L}$. Moreover, the two \emph{tails}, that is, the rays not containing crossing points,
of each line in $\mathcal{L}$ are represented in $A_\mathcal{L}$ by vertices of degree~$1$
adjacent to the crossing points from which the tails emanate.
The graph $A_\mathcal{L}$ is endowed in a natural way with a \emph{path factorization}
by which we mean a partition of the edge set of $A_\mathcal{L}$ into~$k$ paths such that
every two paths have at most one common vertex.

Let $\Gamma$ be a drawing of a graph $G$ on $\mathcal{L}$.
If we see $A_\mathcal{L}$ in a natural way as a topological graph, then
any vertex of $G$ in $\Gamma$ either is placed at a vertex of $A_\mathcal{L}$
or subdivides an edge of $A_\mathcal{L}$. This motivates the following definition.
A \emph{combinatorial description of a $k$-line drawing} of $G$ consists of
\begin{itemize}
\item
a graph $A$ that is the augmented arrangement graph of an arrangement $\mathcal{L}$ of $k$ lines,
\item
the path factorization of $A$ determined by $\mathcal{L}$,
\item
a subdivision $A'$ of $A$,
\item
an isomorphism from $G$ to a subgraph of $A'$.
\end{itemize}

\begin{theorem}\label{thm:fpt-main^2}
  For~$d \in \{2,3\}$, given any graph~$G$ with~$n$ vertices and~$m\ge1$
  edges, and any positive integer $k$,
  the decision problem~\enquote{$\rho^1_d(G)\le k$?} can be solved
  in~$k^{O(k^2)} + O(n+m)$ time.
  If the answer is yes,
  a combinatorial description of a
  $k$-line drawing of $G$ in~$\reals^d$
  can be found within the same time bound.
\end{theorem}

\begin{proof}
  We apply Corollary~\ref{cor:line-cover-fpt}.  If the answer to the
  decision algorithm is no, we are done.  Otherwise, we can assume
  that~$G$ satisfies some of the necessary properties observed in the
  proof of Theorem~\ref{thm:line-cover-kernel}.  For example,
  as in the proof of Theorem~\ref{thm:line-cover-kernel}, we can
  assume that the input graph~$G$ has no path components.  (If there
  are path components, we can place them on the ``far end'' of some
  line.)  For $i<n$,
  we denote by~$V_i$ the set of vertices of a graph $G$ having degree
  $i$ and by $V_{\ge i}$ the set of vertices of $G$ having degree at
  least~$i$.  Again, following the proof of
  Theorem~\ref{thm:line-cover-kernel}, a necessary condition for
  $\rho^1_d(G)\le k$ is
\begin{equation}
  \label{eq:V13}
  |V_{\ge3}|\le{k\choose 2}.
\end{equation}
We also assume this condition here.  Due to the fact that we have no
path components, $|V_1|$ is bounded by the total degree
of the vertices in~$V_{\ge 3}$, which is bounded by~$2k^2$ as we
showed in the proof of Theorem~\ref{thm:line-cover-kernel}.
We call a path in the input graph~$G$ \emph{straight} if its end
vertices do not have degree~$2$ whereas all internal vertices have
degree~$2$.  Similarly, we call a cycle in~$G$ \emph{straight} if all
but one of its vertices have degree~$2$. (A cycle in~$G$ where all
vertices have degree~$2$ is necessarily a component of~$G$.)
Note that $G$ contains fewer than $2k^2$ straight paths and cycles.
(This is due to the fact that each of them ``consumes'' at least one
unit from the sum of the vertex degrees
of the augmented arrangement
graph~$A_\mathcal{L}$ of a $k$-line arrangement~$\mathcal{L}$, which
has at most $k^2$ edges.)

Obviously, all vertices in~$V_{\ge 3}$ are located at crossing
points of a $k$-line arrangement~$\mathcal{L}$.  Let
$S \subseteq V_2$ be the set of degree-2 vertices located at
crossing points of~$\mathcal{L}$.  Clearly,
\begin{equation}
  \label{eq:admissible}
|S|\le{k\choose 2}.
\end{equation}
Note also that $S$ must contain at least two vertices of each straight cycle.
Any such set $S\subseteq V_2$ satisfying also condition~(\ref{eq:admissible})
will be referred to as \emph{admissible}.
Let $G_S$ denote the graph obtained from $G$ by smoothing out all vertices in $V_2 \setminus S$.
More precisely, note that each connected component in
the subgraph induced by $V_2 \setminus S$ is a path incident to one or
two unique vertices in $V_1\cup S\cup V_{\ge 3}$.
We obtain~$G_S$ from~$G$ by replacing each of the paths incident to
two different vertices in $V_1\cup S\cup V_{\ge 3}$ by an edge that
connects the two vertices.

Consider a \emph{slightly subdivided} version of the augmented
arrangement graph~$A_\mathcal{L}$.
Specifically, call an edge \emph{pendant} if it is incident to a vertex of degree~$1$.
For any subset~$E$ of the at most $k(k-1)$ non-pendant edges
of~$A_\mathcal{L}$, we use~$B_{\mathcal{L},E}$ to denote the graph resulting
from~$A_\mathcal{L}$ if we subdivide each edge in~$E$ by two extra
vertices.
The new vertices of $B_{\mathcal{L},E}$ correspond to possible locations
of vertices from $V_1$ between two crossing points on a line in $\mathcal{L}$, whereas edges that have not been subdivided correspond to possible edges between vertices in~$S \cup V_{\ge 3}$.

Note that $G$ is drawable on $\mathcal{L}$ if and only if there exists
an admissible set $S$ and a subset~$E$ of non-pendant edges of~$A_\mathcal{L}$
such that $G_S$ is isomorphic to a subgraph of~$B_{\mathcal{L},E}$.
Our strategy consists, therefore, of the
following steps:
\begin{itemize}
\item for every admissible subset~$S$ of~$V_2$, generate the
  subgraph~$G_S$ of~$G$;
\item generate every possible, up to isomorphism, graph
  $B_{\mathcal{L},E}$ along with the path factorization determined by
  $\mathcal{L}$;
\item for every pair $(G_S,B_{\mathcal{L},E})$, check whether
  $B_{\mathcal{L},E}$ contains a subgraph isomorphic to~$G_S$ and, if
  so, compute an isomorphism $\beta$ from $G_S$ to such a subgraph;
\item
  subdivide $G_S$ back to restore $G$ and extend $\beta$ to an
  isomorphism from $G$ to a subdivided version of $B_{\mathcal{L},E}$;
\end{itemize}
As we show below, we will eventually find a suitable pair of~$G_S$ and
$B_{\mathcal{L},E}$ using our assumption~$\rho^1_d(G) \le k$.

\subparagraph*{Generating $G_S$ for each admissible set~$S$.}
Let $S$ and $S'$ be two subsets of~$V_2$.  We call them
\emph{equivalent} if each straight path and each straight cycle has
equally many vertices in~$S$ and in~$S'$.  Note that, if $S$ and $S'$
are equivalent, then $G_S$ and $G_{S'}$ are isomorphic.

In order to generate representatives of each equivalence class,
we first rename the vertices of~$G$. For the vertices in $V_1$ and $V_{\ge3}$
we use labels of binary length $O(\log k)$; this is possible since
both sets have cardinality $O(k^2)$ (see inequality~(\ref{eq:V13}) and
the argument following it).
The label of a vertex~$v$ of degree~$2$ consists of the labels assigned to the end vertices
of the straight path containing $v$ and the index of $v$ along this
path counted in the direction starting from the end vertex with
lexicographically smaller label. We do similarly for the vertices along a straight cycle,
fixing an orientation of each straight cycle for this purpose.
Let $K={k\choose 2}$.
Now, we generate every subset~$S$ of~$V_2$ with $|S|\le K$
such that the intersection of $S$ with each straight path and cycle is an
\emph{initial} segment of this path or cycle.  This is no restriction
since only the \emph{number} of $S$-vertices in each straight path or
cycle matters.  The initial segment may be empty for
some paths and must contain at least two vertices for every cycle.
Let $\sigma$ denote the number of such subsets of~$V_2$.
Then $\sigma \le {K+2k^2\choose K}$ since the right side bounds
the number of ways of putting the $K$ elements of~$V_2$ into~$2k^2+1$
containers corresponding to the straight paths and
cycles.   (We put the elements of~$V_2$
that are not used by~$S$ into an extra ``dummy''
container.)  Using the bound ${a\choose b}<(e\,a/b)^b$, we get
$\sigma < \big((5+4/(k-1))e\big)^K$.  Assuming $k \ge 2$ yields
$\sigma < (9e)^K <5^{k^2}$ since $\sqrt{9e}<5$.

Let $\mathcal F$ be the set of the graphs~$G_S$ for
every admissible subset~$S$ of~$V_2$.
Clearly, $|\mathcal F|=\sigma<5^{k^2}$.
Moreover, every graph $F$ in $\mathcal F$ has at most
$|V_1|+|S|+|V_{\ge3}|\le 3k^2$ vertices,
and each vertex is represented by a binary label of length $O(\log k)$.
As we observed above, each element $G_S$ of~$\mathcal{F}$ corresponds
to a partition of~$S$ into at most $2k^2+1$ containers where each
container (except the dummy container) gets an initial segment of the
corresponding straight path or cycle.  We can hence enumerate the partitions
of~$S$ with a delay of $O(k^2)$: we index the containers with numbers
$1,2,\dots,2k^2+1$, and then we go through the partitions of~$S$ in
lexicographic order, always making sure that each cycle container gets
at least two elements.
In a preprocessing step, we compute from~$G$ (in $O(n+m)$ time) a
``frame'' so that we can generate~$G_S$ for each admissible set~$S$ in
$O(k^2)$ time: we simply put the vertices in each
container (except for the dummy container) at the
appropriate spot of the frame, that is, at the beginning of the
corresponding straight path or cycle.
Thus, after a preprocessing step taking~$O(k^2+n+m)$ time, we can
enumerate~$\mathcal{F}$ in $5^{O(k^2)}$ time.

\subparagraph*{Generating $B_{\mathcal{L},E}$.}
In general, a graph $H$ with a specified path factorization
is called \emph{factorized}; the corresponding paths will be referred to
as \emph{path factors} of $H$.
We call a factorized graph $H'$ a \emph{pre-template graph} if
\begin{itemize}
\item
$H'$ has no isolated edge and no vertex of degree~$0$ or~$2$; and
\item
every path factor of $H'$ is a path between two degree-$1$ vertices.
\end{itemize}
A factorized graph $H$ is a \emph{template graph} if there is a
pre-template graph $H'$ and a subset $E$ of the non-pendant edges
of~$H'$ such that $H$ is obtained from $H'$ by subdividing each edge
in~$E$ by two new vertices.  Each path factor of $H'$ is ``prolonged''
to a path factor of $H$ correspondingly.
A factorized graph $H$ is \emph{stretchable in $\reals^d$} if
there is a drawing of $H$ in $\reals^d$ such that every path factor of $H$ lies
on a line and different path factors occupy different lines.

Note that a factorized graph~$H$ with $k$ path factors is a template
graph that is stretchable in~$\reals^d$ if and only if there is a
family~$\mathcal{L}$ of $k$ non-isolated lines in $\reals^d$ and a subset~$E$ of
non-pendant edges of~$A_\mathcal{L}$ such that~$H$ is the
slightly subdivided augmented arrangement graph~$B_{\mathcal{L},E}$.
In order to generate all such graphs, first consider the family~$\mathcal{H}_k$ of all
pre-template graphs with $k$ path factors. Label the path factors by
$1,\ldots,k$.  Given $H'\in\mathcal{H}_k$, label each vertex of~$H'$
with the set of labels of all path factors
to which this vertex belongs.
For a path factor~$i$, let $S_i$ be the sequence of labels of all vertices appearing along that path factor.
The list of the sequences $S_1,\ldots,S_k$ determines $H'$
and allows to reconstruct it up to isomorphism.

Now we argue that $\mathcal{H}_k$ contains at most
$(2^kk!)^k \in k^{O(k^2)}$ factorized graphs.  To this end,
for each sequence $S_i$, we choose a permutation of the
other line labels $1,\dots,i-1,i+1,\dots,k$.  Then, between any two
consecutive line labels in the chosen permutation, we decide whether
to start a new vertex label (which is a set of line labels) or not.
The last positive of these $k$ decisions separates the path factors
that intersect path factor~$i$ from those that do not (that is, path
factor~$i$ intersects between~$1$ and~$k-1$ other path factors).
Moreover, each factorized graph in~$\mathcal{H}_k$ has at most
$2k+{k\choose2}$ vertices and $k^2$ edges.
Hence, $\mathcal{H}_k$ can be generated in time $k^{O(k^2)}$.
A template graph $H$ is stretchable if and only if it is obtained from
a stretchable pre-template graph $H'$.
By Corollary~\ref{cor:PhiGk},
the statement saying that a given factorized
graph~$H'\in\mathcal{H}_k$ is stretchable
can be written in~$O(k^4)$ time
as a sentence $\Phi_{H'}$ in the existential
first-order theory of the reals.  Moreover, $\Phi_{H'}$ uses
$O(k^2)$ first-order variables and involves $O(k^4)$ polynomial inequalities
of bounded degree with bounded coefficients. %
Hence the algorithm of Renegar takes time $k^{O(k^2)}$ to check
whether $\Phi_{H'}$ is valid; see Theorem~\ref{thm:renegar}.
Thus, all stretchable pre-template graphs with $k$ path factors
can be generated in time $k^{O(k^2)}\cdot k^{O(k^2)} \in k^{O(k^2)}$.
From each pre-template graph, we generate up to~$2^{k^2}$ template graphs,
one for each possible subset of non-pendant edges.
Consequently,
all stretchable template graphs with $k$ path factors
and, hence, all slightly subdivided augmented arrangement graphs of
$k$-line families in $\reals^d$
can be generated in time $2^{k^2}\cdot k^{O(k^2)} \in k^{O(k^2)}$.

\subparagraph*{Isomorphism.}
Let~$S$ be any admissible set for~$G$, let~$\mathcal{L}$ be any family of~$k$
non-isolated lines in $\reals^d$, and let~$E$ be any subset
of non-pendant edges of~$A_\mathcal{L}$.
Recall that $|V(G_S)|<3k^2$ and note that also $|V(B_{\mathcal{L},E})|<3k^2$.
Thus, by trial of all injections from $V(G_S)$ to
$V(B_{\mathcal{L},E})$, we can check in time $k^{O(k^2)}$ whether
there exists an isomorphism~$\beta$ from~$G_S$ to a subgraph of
$B_{\mathcal{L},E}$.
If it exists,
we can extend~$\beta$ to an isomorphism from~$G$ to a
subgraph of a subdivision of~$B_{\mathcal{L},E}$ using $O(n+m)$ additional time.
Indeed, we just have to restore the vertices in $V_2\setminus S$
that we removed from $G$, along with the incident edges, and
subdivide~$B_{\mathcal{L},E}$ correspondingly.
In other words, for each straight path or cycle,
we have to restore the final segment that was not included in~$F$.
Let~$P$ be a straight path, let~$w$ be the lexicographically
larger end vertex of~$P$, and,
among the vertices of~$P$ that were included in~$F$,
let~$u$ be the neighbor of~$w$ in~$F$.
All what we have to do is to
subdivide $uw$ by the vertices that were removed from~$P$.
We treat the straight cycles similarly.
Lastly, we restore the original vertex names of~$G$ from the
respective vertex labels.

\subparagraph*{Overall time estimation.}
As we have seen, the set~$\mathcal{F}$ (containing~$G_S$ for every
admissible set~$S$) can be enumerated in time~$5^{O(k^2)}$ time after
a preprocessing step requiring~$O(k^2+n+m)$ time.  We have also seen
that the set of all slightly subdivided augmented arrangement graphs
of all families of~$k$ non-isolated lines in $\reals^d$ can
be enumerated in~$k^{O(k^2)}$ time.  Thus, in
time~$k^{O(k^2)}+O(n+m)$, we can generate
the set
that, for every admissible set~$S$, every family~$\mathcal{L}$ of $k$
non-isolated lines in $\reals^d$, and every subset~$E$ of
non-pendant edges of~$A_\mathcal{L}$, contains the
pair~$(G_S,B_{\mathcal{L},E})$.
As discussed above, in time~$k^{O(k^2)}$ we can check whether there is an isomorphism from~$G_S$ to a subgraph of $B_{\mathcal{L},E}$, and, if it exists, we can further extend it in time~$O(n+m)$ to an isomorphism from~$G$ to a subdivision of that subgraph.
This results in an overall running time of $k^{O(k^2)} + O(n+m)$.
\end{proof}

\section{Rational (Non-) Realizability of
  $\rho^1_2$-Optimal Drawings}
\label{sec:irrational-coordinates2}

In this section, we show that there are graphs all of whose
$\rho^1_2$-optimal drawings require a vertex to be mapped to a point
with an irrational coordinate; therefore such a drawing does not fit
on any integer grid.

A \emph{collinearity configuration} is a set $V$ of \emph{abstract points}
along with a family of 3-element subsets of $V$ called \emph{collinear triples}.
A \emph{realization} of the collinearity configuration is an injective mapping~$\alpha\function V{\reals^2}$
such that any three abstract points $a$, $b$, and $c$ form a collinear triple if
and only if the points $\alpha(a)$, $\alpha(b)$, and $\alpha(c)$ are collinear.
The \emph{Perles configuration} is a collinearity configuration of 9 points
whose realization is shown in Fig.~\ref{fig:perles}.
It is known that every realization of the Perles configuration
contains a point with an irrational coordinate~\cite[page 23]{Berger2010}.

    \begin{figure}[htb]
        \begin{subfigure}{.48\textwidth}
            \centering
            \includegraphics[page=2]{perles-configuration}
            \caption{The Perles configuration has at least one vertex with
                an irrational coordinate.}
            \label{fig:perles}
        \end{subfigure}
        \hfill
        \begin{subfigure}{.48\textwidth}
            \centering
            \includegraphics[page=1]{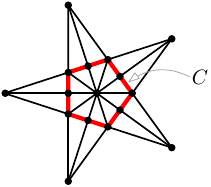}
            \caption{The graph~$G$ contains the Perles configuration in a
                $\rho^1_2$-optimal drawing.}
            \label{fig:g-unb}
        \end{subfigure}
        \caption{The Perles configuration and a supergraph of it.}
        \label{fig:perles+co}
    \end{figure}

\begin{theorem}
    \label{thm:irrational-coord-graph}
    There exists a graph~$G$ such that every drawing~$\Gamma$ realizing~$\rho^1_2(G)$
    contains a vertex with an irrational coordinate.
\end{theorem}

\begin{proof}
Consider the graph~$G$ whose drawing is shown in Fig.~\ref{fig:g-unb}.
Let $r=\rho^1_2(G)$ and note that $r\le 10$ as the drawing of $G$ in Fig.~\ref{fig:g-unb} occupies $10$ lines.

Let~$\Gamma$ be a drawing of~$G$ realizing~$r$.
By the detailed analysis below we will show that~$\Gamma$ has the same
    configuration as the drawing in Fig.~\ref{fig:g-unb}, that is, any three vertices
    are collinear in~$\Gamma$ if and only if they are collinear in
    Fig.~\ref{fig:g-unb}.  Since Fig.~\ref{fig:g-unb} contains a
    realization of the Perles configuration, $\Gamma$
    must contain at least one vertex with an irrational coordinate.

Classify the vertices of $G$ into three classes according to their degrees: one \emph{big vertex} $b$ of degree~$10$,
ten \emph{medium vertices} of degree~$5$ or~$4$, and five \emph{small vertices} of degree~$3$.
Let $m_1,\dots, m_{10}=m_0$ be the medium vertices enumerated as they appear along the pentagon
in Fig.~\ref{fig:g-unb}. They form a cycle of length 10, which we denote by~$C$.

For convenience we shall sometimes identify the vertices or edges of
the graph $G$ with their images in the drawing~$\Gamma$.

For every two
consecutive medium vertices~$m_i$ and~$m_{i+1}$, we claim that there is
no medium vertex~$m_j$ such that the open segment
$(m_i, m_{i+1})$ intersects the ray~$bm_j$.
Indeed, assume the contrary.
Since the points $b$, $m_i$, and $m_{i+1}$ are the vertices of a triangular face,
they cannot be collinear, so the ray $bm_j$ intersects $(m_i, m_{i+1})$ at a single point,
which we denote by $m_j'$. Since the drawing $\Gamma$ is crossing free, the point $m_j$
is closer to the point $b$ than the point $m_j'$. Therefore,
the triangle $\triangle bm_im_{i+1}$
contains the point $m_j$ in its interior.
Since $\Gamma$ is crossing-free, all the remaining vertices of $C$ also lie in the interior
of $\triangle bm_im_{i+1}$.
It follows that the ten lines $bm_1,bm_2,\dots,bm_{10}$ are pairwise
different, and each of them is different from the line~$m_im_{i+1}$.
So we need at least $11$ lines to cover~$\Gamma$; a contradiction.

Thus, the medium vertices appear in $\Gamma$ around the point $b$ in
their innate cyclic order along~$C$.
In particular, the point $b$ belongs to the part $B$ of the plane bounded by $C$, and
the segments $[b,m_i]$, $1\le i\le10$, partition $B$ into ten triangles
$\triangle bm_0m_1, \triangle bm_1m_2, \dots, \triangle bm_9m_{10}$.
It follows that, for each line $\ell$ passing through the
point $b$, the intersection $\ell\cap B$ is a segment.

\begin{claim}\label{cl:small-B}
All small vertices lie outside~$B$.
\end{claim}

\begin{subproof}
  Assume that a small vertex $s$ is drawn in~$B$.  Then $s$ is
  contained inside of a triangle $\triangle bm_im_{i+1}$ for some~$i$.
  However, in this case $s$ can be adjacent only to $m_i$ and
  $m_{i+1}$, and to no third medium vertex; a contradiction.
\end{subproof}

A cover $\mathcal{L}$ of $\Gamma$ by $r$ lines consists of the following three families.
\begin{itemize}
\item $\mathcal{L}_b$, consisting of the lines in $\mathcal{L}$
covering an edge $bm_i$ for some $1\le i\le10$.  Set $r_b=|\mathcal{L}_b|$
and note that $r_b\ge 5$.
\item $\mathcal{L}_m$, consisting of the lines in $\mathcal{L}$ covering the cycle $C$.
It is easy to see that $\mathcal{L}_m\subseteq\mathcal{L}\setminus\mathcal{L}_b$
and, hence, $r_m=|\mathcal{L}_m|\le r-r_b\le 5$.
\item $\mathcal{L}_s=\mathcal{L}\setminus(\mathcal{L}_b\cup \mathcal{L}_m)$.
  We set $r_s=|\mathcal{L}_s|$ and call the lines in
$\mathcal{L}_s$ \emph{special}.  An edge of~$G$ that is not covered by
any line in $\mathcal{L}_b\cup \mathcal{L}_m$ will be called \emph{special}, too.
Thus, a special edge has to be covered by a special line.
Note also that every special edge connects
a small and a medium vertex.
\end{itemize}

\begin{claim}\label{cl:L_b}
If~$s$ is a small vertex, then $\mathcal{L}_b$ covers at most one edge incident to~$s$. %
\end{claim}

\begin{subproof}
Assume that a line $\ell\in\mathcal{L}_b$ covers an edge incident to $s$.
Since $\ell$ passes through the points $s$ and $b$, $\ell$ is unique.
Recall that the intersection $\ell\cap B$ is a segment.
By Claim~\ref{cl:small-B}, the point $s$ lies outside $B$.
Since  $s$ is adjacent only to medium vertices, all lying on $C\subset B$, the line
$\ell$ covers no other edge incident to~$s$.
\end{subproof}

The cycle $C$ is drawn as a closed polyline $\hat C$
with $3\le r_C\le 10$ corners.
Therefore, $3\le r_m\le 5$ and so $r_s\le r-r_b-r_m\le 2$.

\begin{claim}\label{cl:mostone}
If $r_m\le 4$, then every line in $\mathcal{L}_m$ covers exactly one side of $\hat C$, and $r_C=r_m$.
\end{claim}
\begin{subproof}
If a line $\ell\in  \mathcal{L}_m$ covers two sides of $\hat C$, then each of the four endpoints of these sides need a separate line from $\mathcal{L}_m$ to cover its incident side not covered by $\ell$.
\end{subproof}

\begin{claim}\label{cl:L_m}
Suppose that $r_m\le 4$. If $s$ is a small vertex, then a line
$\ell\in\mathcal{L}_m$ covers at most one edge incident to~$s$.
\end{claim}

\begin{subproof}
By Claim~\ref{cl:mostone}, $r_C=r_m\le 4$.
This bound implies that
the intersection $\ell\cap B$ is a segment.
If $\ell$ passes through $s$, then $\ell$ can contain
at most one of the three medium vertices adjacent to $s$, because all of them lie on $\hat C$
whereas $s$ lies outside $B$ by Claim~\ref{cl:small-B}.
\end{subproof}

\begin{claim}\label{cl:r_m3}
$r_m>3$.
\end{claim}

\begin{subproof}
Suppose for a contradiction that $r_m=3$. Then Claim~\ref{cl:mostone} implies $r_C=3$,
and the family $\mathcal{L}_m$ consists of the three lines
containing the sides of $\hat C$. If a small vertex $s$ belongs to two lines from the
family $\mathcal{L}_m$, then $s$ coincides with one of the
corners of $\hat C$, which is impossible. Therefore, $s$ is covered by at most one
line from $\mathcal{L}_m$. By Claim~\ref{cl:L_m}, this line can cover at most one edge incident to $s$.
Along with Claim~\ref{cl:L_b} this implies that the family $\mathcal{L}_b\cup \mathcal{L}_m$
covers at most two edges incident to $s$. Thus, for each small vertex $s$, there exists
at least one edge incident to $s$ that has to be covered by a special line; let $e(s)$ be one such edge.
Since we have $r_s\le 2$ special lines covering five special edges $e(s)$,
there exists a special line $\ell$ containing at least three of these edges,
say $e(s_1)$, $e(s_2)$, and $e(s_3)$.
The three points $s_1$, $s_2$, and $s_3$ split the line~$\ell$ into
four parts~$\ell^1, \ell^2, \ell^3, \ell^4$.
Since the triangle $B$ is a convex set that, by Claim~\ref{cl:small-B}, contains no small vertex,
it can intersect only
one of these parts, say $\ell^j$. The part $\ell^j$ has one or two
endpoints, say, $s_1$, or~$s_1$ and~$s_2$.
Then the edge $e(s_3)$, connecting $s_3$ with $C$ along $\ell$, must contain one of the
vertices $s_1$ and $s_2$,
a contradiction.
\end{subproof}

\begin{claim}\label{cl:r_m4}
$r_m>4$.
\end{claim}

\begin{subproof}
Suppose for a contradiction that $r_m=4$. Then by Claim~\ref{cl:mostone} $r_C=4$
 and the family $\mathcal{L}_m$ consists of the four lines
containing the four sides of quadrilateral~$\hat C$.

First, assume that $\hat C$ is nonconvex.
Note that
$B$ contains all crossing points of the lines from~$\mathcal{L}_m$.
So each small vertex
$s$ can be covered by at most one line from $\mathcal{L}_m$.
Recall that, by Claim~\ref{cl:L_m}, such a line can cover at most one edge incident to $s$.
Together with Claim~\ref{cl:L_b} this implies that the family~$\mathcal{L}_b\cup \mathcal{L}_m$
covers at most two edges incident to $s$. Thus, for each small vertex $s$, there exists an
edge~$e(s)$ incident to $s$ that has to be covered by a special line.
Since $r_s\le r-r_b-r_m\le 1$, there is a single special line $\ell$, which covers all five special edges $e(s)$.
The five small points split~$\ell$ into six parts.
The quadrilateral $B$ consists of two triangles, that contain no small vertices by Claim~\ref{cl:small-B}.
It follows that $B$ can intersect at most two of the six parts of~$\ell$. These two parts
have at most four endpoints, and hence there is a small vertex $s'$ different from them.
Since the edge $e(s')$ contains no small vertex except $s'$ and, in particular, no endpoint of
the two intersected parts, it cannot reach a medium vertex on $\hat C$ along $\ell$; a contradiction.

Now assume that $\hat C$ is convex.
The four lines from the family $\mathcal{L}_m$ (which are the straight-line extensions of the
sides of $B$) have at most ${4 \choose 2}=6$ crossing points.
Four of them are vertices of~$B$.  Therefore, at most two small
vertices can be crossing points of two
lines from $\mathcal{L}_m$. Let $s_1$, $s_2$, and $s_3$ be the three remaining small vertices.
Thus, each $s_i$ can be covered by at most one line from $\mathcal{L}_m$,
which, by Claim~\ref{cl:L_m}, covers at most one edge incident to $s_i$.
Applying Claim~\ref{cl:L_b}, we have that the family $\mathcal{L}_b\cup \mathcal{L}_m$
covers at most two edges incident to $s_i$. Thus, for each $i=1,2,3$, there exists an
edge $e(s_i)$ incident to $s_i$ that has to be covered by a special line.
Since there is a single special line $\ell$, it covers all three special edges $e(s_i)$.
The three points $s_1$, $s_2$, and $s_3$ split $\ell$ into four
parts~$\ell^1, \ell^2, \ell^3, \ell^4$.
Since $B$ is a convex quadrilateral that contains no small vertex (by Claim~\ref{cl:small-B}),
$B$ can intersect only
one of these parts, say $\ell^j$. One of the three small vertices,
say $s_1$,
is different from any endpoint of
$\ell^j$. The edge $e(s_1)$ cannot reach $\hat C$ along $\ell$ without crossing
another small vertex, which is an endpoint of $\ell^j$;
a contradiction.
\end{subproof}

Therefore, $r_m=5$. Thus, $r_s=0$, $\mathcal{L}_s=\emptyset$,
and the drawing $\Gamma$ is covered by $\mathcal{L}_b\cup \mathcal{L}_m$.
Since each small vertex $s$ has degree~$3$, this vertex is the
crossing point of at least
two lines of $\mathcal{L}$ in the drawing $\Gamma$.
We say a vertex $s$ is \emph{proper} if two of these lines are in
$\mathcal{L}_m$ and \emph{improper}, otherwise.

In the latter case, by Claim~\ref{cl:L_b}, $s$ is a crossing point of
exactly two lines of $\mathcal{L}$, namely, of a line
$\ell\in\mathcal{L}_m$ and a line $\ell'\in\mathcal{L}_b$, where
$\ell$ covers two edges incident to $s$ and $\ell'$ covers one such edge.
Let $s$ be adjacent to the medium vertices $m_{i-1}$, $m_i$, and
$m_{i+1}$. Since triangles $\triangle sm_{i-1}m_i$ and $\triangle sm_im_{i+1}$ share the edge $sm_i$, we see that the vertices
$m_{i-1}$, $s$, and $m_{i+1}$ lie, in this order,
on the line $\ell$. Since the point $m_i$ lies on $\ell'$ between $b$ and $s$,
it is an interior point of the triangle $\triangle bm_{i-1}m_{i+1}$.
That is, $s$ is placed outside the reflex angle of the polygon~$\hat C$.
Since this angle is determined by the triple of vertices $m_{i-1}$, $m_i$, and $m_{i+1}$,
every reflex angle of $\hat C$ hosts no more than one small vertex.

We claim that $\hat C$ is a pentagon.
To show this, consider $\mathcal{L}_m$ as a set of vertices of an auxiliary graph such that
lines $\ell$ and $\ell'$ of $\mathcal{L}_m$ are adjacent if and only
if the crossing point of $\ell$ and $\ell'$ is a vertex of~$\hat C$.
Going along the boundary of $\hat C$, and listing the lines of
$\mathcal{L}_m$ covering passed edges, we construct an Eulerian cycle
in $\mathcal{L}_m$.
Since each vertex of the auxiliary graph has degree at least two and at most four,
it is easy to check that the auxiliary graph is Eulerian if and only if the multiset of degrees of vertices
of the graph $\mathcal{L}_m$ is one of the following:
$\{2,2,2,2,2\}$, $\{4,2,2,2,2\}$, $\{4,4,2,2,2\}$, or $\{4,4,4,4,4\}$.
In the first case, $\hat C$ is a pentagon. We shall show that the other cases are impossible.

\begin{itemize}
\item Case $\{4,2,2,2,2\}$. In this case $\hat C$ is a hexagon and
  there is one line in $\mathcal{L}_m$, covering two sides of $\hat
  C$, and four lines, covering one side of $\hat C$. So it is easy to
  study all possible shapes for $\hat C$, see
  Fig.~\ref{fig:Th7_Hexagons}.  It is easy to check that each of the
  shapes have at most two reflex angles (for placing an improper
  vertex) and that in each of the shapes there are at most two
  crossing points of two lines of $\mathcal L_m$ outside $\hat C$, so
  at most two proper vertices can be drawn.  As there are five small
  vertices, we get a contradiction.
\begin{figure}
\centering
\includegraphics{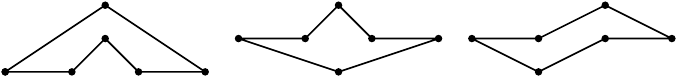}
\caption{Possible shapes for a hexagon  $\hat C$ when  $r_m=5$.}
\label{fig:Th7_Hexagons}
\end{figure}

\item Case $\{4,4,2,2,2\}$. In this case $\hat C$ is a heptagon and there are two (necessarily intersecting) lines in $\mathcal{L}_m$, covering two sides of $\hat C$, and three lines, each covering one side of $\hat C$. So it is easy to study all possible shapes for $\hat C$, see Fig.~\ref{fig:Th7_Heptagons}.
\begin{figure}
\centering
\includegraphics{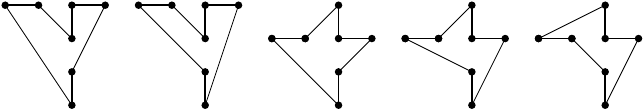}
\caption{Possible shapes for a heptagon  $\hat C$ when  $r_m=5$.}
\label{fig:Th7_Heptagons}
\end{figure}

It is easy to check that each of the shapes has at most three reflex
angles (for placing an improper vertex) and that in each of the shapes
there is at most one crossing point of two lines of~$\mathcal L_m$
outside $\hat C$, so at most one proper vertex can be drawn.  As there
are five small vertices, we get a contradiction.

\item Case $\{4,4,4,4,4\}$. In this case $\hat C$ is a decagon. For
  each line $\ell\in\mathcal L_m$, the endpoints of sides of $\hat C$
  covered by $\ell$ are its crossing points with the other lines of
  $\mathcal{L}_m$. Since there are at most four such points, $\hat C$
  can be a decagon only if each line $\ell\in\mathcal L_m$ intersects
  all other lines of $\mathcal{L}_m$ and all its crossing points are
  endpoints of two sides of $\hat C$.  Since each improper vertex
  belonging to $\ell$ is adjacent to the endpoints of such distinct
  sides, $\ell$ covers at most one improper vertex $s$.  By the
  properties of an improper vertex, $\ell$ covers two edges incident
  to~$s$.  Let~$s$ be adjacent to the medium vertices~$m_{i-1}$,
  $m_i$, and~$m_{i+1}$.  Then the vertices $m_{i-1}$, $s$, and
  $m_{i+1}$ lie, in this order, on a line $\ell$.  Note that the
  angle~$\angle m_i$ at~$m_i$ in the decagon~$\hat C$ is reflex.
Thus, $s$ contributes
\begin{align*}
  &\hspace{-5ex}\angle m_{i-1}/2+\angle m_{i}+\angle m_{i+1}/2= \\
  &= \frac{\pi-\angle m_im_{i-1}m_{i+1}}{2} +
    (2\pi - \angle m_{i-1}m_im_{i+1}) + \frac{\pi-\angle m_im_{i+1}m_{i-1}}{2} \\
  &= 3\pi-
    \frac{\angle m_im_{i-1}m_{i+1}+\angle m_{i-1}m_im_{i+1}+\angle m_im_{i+1}m_{i-1}}{2}
    - \frac{\angle m_{i-1}m_im_{i+1}}{2} \\
  &= \frac{5\pi}{2}-\frac{\angle m_{i-1}m_im_{i+1}}{2} > 2\pi
\end{align*}
to the sum of angles of $\hat C$.  Since each crossing point among the
lines of $\mathcal L_m$ is an endpoint of a side of the decagon~$\hat
C$, there remain no free crossing points to place proper vertices,
so all small vertices are improper. Each of them contributes more than
$2\pi$ to the sum of angles of the decagon, which equals $8\pi$; a contradiction.
\end{itemize}

That is, $\hat C$ is a pentagon.
We claim that $B$ is convex.
Suppose for a contradiction that $B$ is nonconvex.
Then $\hat C$ has one or two reflex angles. Let $A$ be any of them.
Each of the two lines of~$\mathcal L_m$
covering the sides of $A$ has at most one crossing point with other lines
of $\mathcal L_m$ outside $B$.  It follows that at most three crossing
points of lines of~$\mathcal L_m$ are outside $B$.
Therefore at least two small vertices are improper.
Let $s$ be one of them and $s\in\ell\in \mathcal L_m$. Let $m_{i-1}$ and $m_{i+1}$ be
neighbors of~$s$ placed on~$\ell$.  It is easy to see (for instance,
in Fig.~\ref{fig:hatCtotriangle}) that we transform~$\hat{C}$ into
a triangle if we replace the line segments $[m_{i-1},m_i]$ and
$[m_i,m_{i+1}]$ by the segment $[m_{i-1},m_{i+1}]$.  It follows that
$\hat C$ had only one reflex angle, so we could place only one improper 
small vertex; a contradiction.

\begin{figure}
  \centering
  \includegraphics{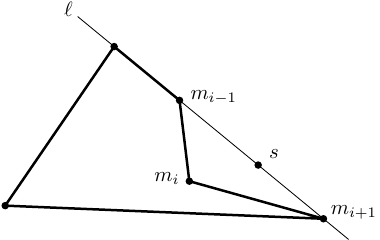}
  \caption{The transformation of the pentagon $\hat C$ into a triangle.}
  \label{fig:hatCtotriangle}
\end{figure}

Thus, $B$ is a convex pentagon. Since an improper small vertex
has to be placed outside a reflex angle of $\hat C$, every small vertex is proper.
Therefore, each such vertex is placed at the crossing point of two lines
from $\mathcal{L}_m$, which extend two non-adjacent sides of $\hat C$.
Since $\hat C$ is a pentagon, it has has exactly five pairs of non-adjacent sides.
Therefore, the extensions of the sides in each pair intersect and, moreover,
the crossing point lies outside~$B$ by Claim~\ref{cl:small-B}.
Since $B$ is a convex pentagon, no three lines from $\mathcal{L}_m$ can have a common point.
Therefore, it is impossible that each edge incident to a small vertex $s$ is covered
by its own line in $\mathcal{L}_m$. Since $B$ is convex, similarly to the proof of
Claim~\ref{cl:L_b} we can show that each line from $\mathcal{L}_m$ covers at most one edge incident to $s$.
It follows that each small vertex $s$ is the crossing point of exactly two lines from $\mathcal{L}_m$
and exactly one line in $\mathcal{L}_b$, which we denote by~$\ell_s$.

The equality $r_m=5$ implies that $r_b=5$ (and, hence, $r=10$).
Since each of the five lines in $\mathcal{L}_b$ can intersect $\hat C$ in at most two points
and there are 10 medium vertices, each $\ell\in\mathcal{L}_b$ intersects $\hat C$ in exactly two points,
representing two medium vertices $m(\ell)$ and $m'(\ell)$.

The convexity of $B$ also implies that this pentagon is contained in the angle created by the
extensions of any pair of its non-adjacent sides.
Therefore, each side $l$ of $B$ is contained in the triangle $T$
created by the line extending~$l$ and the two non-adjacent sides of~$B$.
Thus, for each side $l'$ of $B$ adjacent to $l$, the triangle $T'$ cut
from $T$ by $l'$ shares with $B$ only the side $l'$.
Let $s$ be the vertex of $T'$ opposite to its side $l'$.
Note that $s$ is a small vertex. The line $\ell_s\in\mathcal{L}_b$ crosses the boundary
$\hat C$ of $B$ at the points $m(\ell_s)$ and $m'(\ell_s)$, and one of these points
must be an inner point of $l'$. Thus, every side of $B$ contains a medium vertex as
an inner point. Also, the five corners of~$\hat C$ must be medium vertices.
It readily follows that, for each $\ell\in\mathcal{L}_b$
(that crosses~$b$ by definition
and contains one small vertex each), one of the medium points $m(\ell)$ and $m'(\ell)$
is a vertex of $B$ and the other one is an inner point of the opposite side of~$B$.

We now see that, as claimed, three vertices of $G$ are collinear in~$\Gamma$ if and only if
they are collinear in the drawing in Fig.~\ref{fig:g-unb}.
\end{proof}

\section{Computational Complexity of the Plane Cover Number}
\label{sec:complexity-rho23}

While the graphs with $\rho_3^2$-value~1 are exactly the planar graphs
(which can be recognized in polynomial time), recognizing graphs with
$\rho_3^2$-value $k$, for any $k > 1$, immediately becomes \cNP-hard.
We prove this via the \cNP-hardness of a new problem of \PTSAT
type, which we think is of independent interest.

\begin{definition}
    Let $\Phi$ be a Boolean CNF formula where every clause consists of
    at most three literals.
    The \emph{associated graph} (also known as the
    \emph{variable--clause incidence graph}) of~$\Phi$, $G(\Phi)$, has
    a vertex~$v_x$ for each variable~$x$ in~$\Phi$ and a vertex~$v_c$
    for each clause~$c$ in~$\Phi$.  There is an edge between a
    variable-vertex $v_x$ and a clause-vertex~$v_c$
    if and only if $x$ or $\neg x$ appears in $c$.
    The Boolean formula $\Phi$
    is called \emph{planar} if
    $G(\Phi)$ is planar.  In the problem \PTSAT, the task is to decide
    whether a given planar 3-CNF formula
    has a satisfying truth assignment.
\end{definition}

Kratochv\'il et al.~\cite{kln-nctg-SIAMJDM91} proved the \cNP-hardness of
\PCTSAT, which is a variant of \PTSAT restricted to instances where
$G(\Phi)$ can be augmented by a simple cycle~$C$ that goes exactly
through the clause-vertices such that $G(\Phi)+C$ remains planar.
In their reduction every clause consists of at least two variables.

Mulzer and Rote~\cite{mr-mwtnp-JACM08}
proved the \cNP-hardness of \PPOITSAT, another variant of
\PTSAT where all literals are positive and the assignment must
be such that, in each clause, exactly one of the three variables is
true.  We combine ideas from the two proofs to show the \cNP-hardness
of the following new problem.

\begin{definition}
  In the \PPCOITSAT problem, we are given a collection~$\Phi$ of
  clauses each of which consists of exactly three literals and each
  literal is positive, together with a planar embedding of
  $G(\Phi) + C$ where $C$ is a cycle that goes exactly through all
  clause-vertices.
  The problem is to decide whether $\Phi$ admits a truth assignment
  such that exactly one variable in each clause is true.
\end{definition}

\begin{lemma}\label{lem:PosPlaCycOneNPC}
  \PPCOITSAT is \cNP-complete.
\end{lemma}

\begin{proof}
  It is easy to see that the problem is in \cNP as it is a
  constrained version of the \cNP-complete \PPOITSAT problem.
  To show \cNP-hardness we reduce from \TSAT and use the construction by
  Kratochv\'il et al.~\cite{kln-nctg-SIAMJDM91} to get an equivalent instance
  of \PCTSAT (represented by a formula $\Phi$) together with a cycle
  $C$ through the clause-vertices and a planar embedding $\Gamma$ of the graph
  $G(\Phi) + C$.
  Note that we do not reduce directly from \PCTSAT as the input of this problem
  does not specify a cycle through the clause-vertices---it only promises its
  existence.
  Remember that each clause contains at least two variables in the
  construction of Kratochv\'il at al.

  In the following,
  we iteratively replace the clauses in $\Gamma$ by positive \OITSAT
  clauses while maintaining the cycle through these clauses.
  Hence, we ultimately obtain a \PPCOITSAT instance.
  Our reduction uses some of the gadgets from the proof of Mulzer and
  Rote~\cite{mr-mwtnp-JACM08}.
  We show how to maintain the cycle when inserting these gadgets.
  Some of the gadgets need to be modified to get the cycle in place,
  others can be simplified slightly as we do not insist on a rectilinear
  layout.  Recall that, in \OITSAT, a clause is a triple $(a,b,c)$
  which is satisfied if exactly one of the variables $a$, $b$, or $c$ is true.
  Thus, a clause is not a disjunction of its three literals as in \TSAT.

  We consider the interaction between the cycle and the clauses in $\Gamma$.
  In $\Phi$, every clause consists of two or three literals and thus
  there are two or three faces around a clause in~$\Gamma$.
  There are two options for the cycle: (O1) it can ``touch'' the clause, that
  is, the incoming and the outgoing edge are drawn in the same face; (O2) it can
  ``pass through'' the clause, that is, incoming and outgoing edge are drawn in
  different faces.

  First, we describe the inequality gadget, which was also used by Mulzer and
  Rote.
  It enforces that two variables $x$ and $y$ have different truth values
  (denoted by $x \neq y$).
  They implement this gadget as
  \[ (x,a,y) \land (a,b,c) \land (a,c,d) \land (b,c,d), \]
  where $a$, $b$, $c$, and $d$ are new variables that are only used inside this
  gadget.
  It is easy to see that this expression is satisfiable if exactly one of the
  variables $x$ and $y$ is true~\cite[Lemma 3.4]{mr-mwtnp-JACM08}.
  A planar drawing of the gadget including the cycle through the clauses is
  depicted in Fig.~\ref{fig:inequalityGadget}.

\tikzset{
  variableNodes/.append style={
    every node/.append style={draw,circle,fill=white},
    inner sep=0pt,
    minimum size=13pt
  },
  ineqNodes/.append style={
    draw,rectangle,fill=white,
    inner sep=0pt,
    minimum size=13pt
  },
  gadgetBg/.append style={fill=black!15},
  clauseCycle/.append style={dash pattern=on 2pt off 1pt,thick,->},
  clauseCycleCommon/.append style={clauseCycle,color=red!90!black},
  clauseCycleO1/.append style={clauseCycle,color=green!70!black},
  clauseCycleO2/.append style={clauseCycle,color=blue!80!black},
}
\newcommand{\NPdrawclauseconnections}[1]{
  \foreach \clause/\x/\y/\a/\b/\c in #1 {
    \coordinate (coord\clause) at (\x,\y);
    \draw (\a) -- (coord\clause);
    \draw (\b) -- (coord\clause);
    \draw (\c) -- (coord\clause);
  }
}
\newcommand{\NPdrawclauses}[1]{
  \foreach \clause/\null in #1 {
    \node[draw=black,fill=white] (\clause)
      at (coord\clause) {};
  }
}
\newcommand{\NPdrawineqs}[1]{
  \foreach \a/\b in {#1} {
    \draw (\a) -- (\b);
    \node[ineqNodes] (ineq\a\b)
      at ($ (\a)!0.5!(\b) $) {$\neq$};
  }
}

\begin{figure}[tb]
  \begin{subfigure}[b]{0.56\textwidth}
    \normalsize
    \centering
    \begin{tikzpicture}[]
      \coordinate (xc) at (0,-0.5);
      \coordinate (yc) at (4,-0.5);

      \fill[gadgetBg] (xc)
        -- ++(0,1.3) -- ++(4,0)
        -- (yc)
        -- ++(0,-1.7) -- ++(-4,0)
        -- cycle;

      \begin{scope}[variableNodes]
        \node[fill=white] at (xc) (x) {$x$};
        \node[fill=white] at (yc) (y) {$y$};
        \node at (2,-0.25) (a) {$a$};
        \node at (3,-1.75) (b) {$b$};
        \node at (2,-1) (c) {$c$};
        \node at (1,-1.75) (d) {$d$};
      \end{scope}

      \newcommand{\clauses}{%
        axy/2/0.5/a/x/y,
        acd/1/-1/a/c/d,
        abc/3/-1/a/b/c,
        bcd/2/-1.75/b/c/d}
      \NPdrawclauseconnections{\clauses}
      \NPdrawclauses{\clauses}
      \draw[clauseCycleCommon]
        ($ (acd) + (-1.3,-1.105) $) -- (acd);
      \draw[clauseCycleCommon] (acd) -- (bcd);
      \draw[clauseCycleCommon] (bcd) -- (abc);
      \draw[clauseCycleCommon] (abc) -- (axy);
      \draw[clauseCycleO1]
        (axy) -- ++ (-2.5,-2.125)
        node[anchor=south east,xshift=2.5mm,yshift=1mm] {O1};
      \draw[clauseCycleO2]
        (axy) -- ++(-2.5,0)
        node[anchor=south west,blue,xshift=-2mm] {O2};

      \draw[dotted,thick] (x) -- ++(-1,0)
        (y) -- ++(1,0);
    \end{tikzpicture}
    \caption{Gadget of Mulzer and Rote for $x \neq y$.}
    \label{fig:inequalityGadget}
  \end{subfigure}
  \hfill
  \begin{subfigure}[b]{0.43\textwidth}
    \normalsize
    \centering
    \begin{tikzpicture}[]
      \coordinate (xc) at (-2,0);
      \coordinate (yc) at (2,0);

      \fill[gadgetBg] (xc)
        -- ++(1,1.5) -- ++(2,0)
        -- (yc)
        -- ++(-1,-1.5) -- ++(-2,0)
        -- cycle;

      \begin{scope}[variableNodes]
        \node[fill=white] at (xc) (x) {$x$};
        \node[fill=white] at (yc) (y) {$y$};
        \node at (-1,-0.6) (a) {$a$};
        \node at (0,1) (b) {$b$};
        \node at (0,-1) (c) {$c$};
      \end{scope}

      \newcommand{\clauses}{%
        abc/0/0/a/b/c}
      \NPdrawclauseconnections{\clauses}
      \NPdrawclauses{\clauses}
      \NPdrawineqs{b/x, c/y}

      \draw[clauseCycleCommon]
        ($ (ineqbx)!-0.7!(ineqcy) $) -- (ineqbx);
      \draw[clauseCycleCommon]
        (ineqbx) -- (abc);
      \draw[clauseCycleCommon]
        (abc) -- (ineqcy);
      \draw[clauseCycleO1]
        (ineqcy) -- ++(0.9,2)
        node [anchor=north west] {O1};
      \draw[clauseCycleO2]
        (ineqcy)
        -- ($ (ineqcy)!-0.7!(ineqbx) $)
        node [anchor=south west,xshift=-0.2cm, blue] {O2};

      \draw[dotted,thick] (x) -- ++(-1,0)
        (y) -- ++(1,0);
    \end{tikzpicture}
    \caption{Our gadget for the clause $x \lor y$.}
    \label{fig:2clauseGadget}
  \end{subfigure}

  \medskip

  \begin{subfigure}[b]{\linewidth}
    \normalsize
    \centering
    \begin{tikzpicture}[]
      \coordinate (xc) at (0,6);
      \coordinate (yc) at (-6,0);
      \coordinate (zc) at (6,0);
      \fill[gadgetBg] (xc) -- (yc) -- (zc) -- cycle;

      \begin{scope}[variableNodes]
        \node[fill=white] at (xc) (x) {$x$};
        \node[fill=white] at (yc) (y) {$y$};
        \node[fill=white] at (zc) (z) {$z$};
        \node at (-2.2,3) (a) {$a$};
        \node at (-3.2,1.5) (q) {$q$};
        \node at (-1.2,1.5) (b) {$b$};
        \node at (-0.2,3) (u) {$u$};
        \node at (1.8,3) (e) {$e$};
        \node at (1.8,1) (c) {$c$};
        \node at (3.8,1) (d) {$d$};
        \node at (2.8,2) (r) {$r$};
      \end{scope}

      \newcommand{\clauses}{%
        aux/-1.2/4/a/x/u,
        abq/-2.2/2/a/b/q,
        buy/-0.2/1/b/u/y,
        cdr/2.8/1/c/d/r}
      \NPdrawclauseconnections{\clauses}
      \NPdrawclauses{\clauses}
      \NPdrawineqs{u/e, c/e, d/z}
      \draw[clauseCycleCommon]
        ($ (ineqdz) + (1,1) $) -- (ineqdz);
      \draw[clauseCycleCommon]
        (ineqdz) -- ($ (d) + (-0.5,-0.5) $) -- (cdr);
      \draw[clauseCycleCommon] (cdr) -- (ineqce);
      \draw[clauseCycleCommon] (ineqce) -- (ineque);
      \draw[clauseCycleCommon] (ineque) -- (buy);
      \draw[clauseCycleCommon]
        (buy) -- ($ ($ (abq)!0.5!(buy) $) + (-0.4,-0.4) $) -- (abq);
      \draw[clauseCycleCommon] (abq) -- (aux);
      \draw[clauseCycleO1]
        (aux) -- ++(5,0)
        node[anchor=north west,xshift=-0.1cm] {O1};
      \draw[clauseCycleO2]
        (aux) -- ++(-2.5,0)
        node[anchor=north east,xshift=0.1cm] {O2};

      \draw[dotted,thick] (x) -- ++(0,1)
        (y) -- ++(-1,0)
        (z) -- ++(1,0);
    \end{tikzpicture}
    \caption{%
        Gadget of Mulzer and Rote for the clause $x \lor y \lor z$.
    }
    \label{fig:3clauseGadget}
  \end{subfigure}

  \caption{Gadgets for our \cNP-hardness proof.  Variables are
    represented by circles, clauses by empty boxes.  Boxes with
    inequality sign represent the inequality gadget depicted in~(a).
    The dashed line shows how we weave the cycle through the clauses.
    There are two variants of the cycle, which differ in only one
    edge.  The cycle either touches (O1) or passes through (O2) the
    gadget.} %
  \label{fig:np-gadgets}
\end{figure}

  Second, we show how to replace a clause $(x \lor y \lor z)$ consisting of three literals.
  To this end, we use a variant of the gadget of Mulzer and Rote,
  in which we replaced the equality gadget by two inequalities:
  \[ (x,u,a) \land (y,u,b) \land (a,b,q) \land (u \neq e) \land (e \neq c)
    \land (d \neq z) \land (c,d,r). \]
  As before, the variables $a$, $b$, $c$, $d$, $e$, $q$, $r$, and $u$ are new
  variables that are not used outside the gadget.
  This expression is satisfiable if and only if $x \lor y \lor z$
  holds~\cite[Lemma 3.5]{mr-mwtnp-JACM08}.
  Figure~\ref{fig:3clauseGadget} shows the gadget with the cycle through the clauses.
  The four clauses in the inequality gadget can obviously be included on the
  cycle by a short detour.
  If $x$ occurs in a clause as a negative literal, we add a new variable that
  is connected to $x$ by an inequality gadget.
  For all gadgets that we use in our reduction it is possible to add the clauses
  of this inequality gadget into the cycle.

  Finally, we have to consider clauses that consist of only two
  literals: $x \lor y$.  The corresponding construction of Mulzer and Rote did not
  allow us to add a cycle through the clauses.  Therefore, we use a
  new gadget:
  \[ (a,b,c) \land (x \neq b) \land (y \neq c) \]
  We show that the clause in this gadget is satisfiable if and only if $x \lor y$ holds:
  If both $x$ and $y$ are false, $b$ and $c$ are true and hence more than one
  variable in the clause is true.
  If $x$ and $y$ are true, $b$ and $c$ are false and the clause can be
  fulfilled by setting $a$ to true.
  If $x$ and $y$ have distinct truth values, $b$ and $c$ also have distinct
  truth values and thus the clause can be fulfilled by setting $a$ to
  false.
  A drawing of the gadget including the cycle is depicted in
  Fig.~\ref{fig:2clauseGadget}.

  If some of the literals in the \TSAT clause are negated, we can simply add a
  new variable and an inequality gadget. For example, we can transform the clause
  $(\neg x \lor y \lor z)$ to $(x \neq a) \land (a \lor y \lor z)$, where $a$ is
  a new variable.

  In summary we constructed a \PPCOITSAT instance
  that is satisfiable if and only if the given \TSAT formula is satisfiable.
  Of course, this transformation can be done in polynomial time
  and hence we showed \cNP-completeness of the problem.
\end{proof}

We now introduce what we call the \emph{intersection line gadget}; see
Fig.~\ref{fig:intersectionLineGadget}.
It consists of a $K_{3,4}$ in which the vertices in the smaller set of the
bipartition---denoted by $v_1$, $v_2$, and $v_3$---are connected by a path.
The vertices in the other set are denoted by $u_1$, $u_2$, $u_3$, and $u_4$.

\begin{figure}[htb]
  \normalsize
  \centering
  \begin{tikzpicture}
    \coordinate (v1) at (-0.4,0);
    \coordinate (v2) at (0,0);
    \coordinate (v3) at (0.4,0);

    \draw (0,0) ellipse (0.7 and 0.3);

    \coordinate (o1) at (0.65,1);
    \coordinate (o2) at (0.65,-1);
    \coordinate (o3) at (-0.65,-1);
    \coordinate (o4) at (-1,2);
    \foreach \i/\style in {1/,2/,3/} {
      \fill[\style] (o\i) circle(2pt);
      \foreach \k in {1,2,3} {
        \draw[\style] (o\i) -- (v\k);
      }
    }
      \fill[] (o4) circle(2pt);
      \foreach \k in {1,2,3} {
        \draw[thick] (o4) -- (v\k);
      }

    \begin{scope}[]
      \fill (v1) circle(2pt);
      \fill (v2) circle(2pt);
      \fill (v3) circle(2pt);
      \node[anchor=east,xshift=-0.25cm] at (v1) {$v_1$, $v_2$, $v_3$};
      \draw (v1) -- (v2) -- (v3);
    \end{scope}

    \node[anchor=west] at (o4) {$u_1$};
    \node[anchor=west] at (o1) {$u_2$};
    \node[anchor=west] at (o2) {$u_3$};
    \node[anchor=east] at (o3) {$u_4$};

    \begin{scope}[shift={(2,0)},scale=\fignphardnessscale]
      \coordinate (symb) at (0,0);
      \coordinate (b) at (0,1.0);
      \coordinate (c1) at (0.4-0.2,0);
      \coordinate (c2) at (0.4,0);
      \coordinate (c3) at (0.4+0.2,0);
      \draw (c2) ellipse (0.4 and 0.2);
      \foreach \k in {1,2,3} {
        \draw (c\k) -- (b);
        \fill (c\k) circle(1.5pt);
      }
      \draw (c1) -- (c3);
    \end{scope}
    \fill (b) circle(2pt);
    \node[anchor=west] at (b) {$u_1$};

    \draw[->,thick] ($ (v3)!0.3!(c1) $) -- ($ (v3)!0.7!(c1) $);

  \end{tikzpicture}
  \caption{The intersection line gadget and how it is depicted in
      Fig.~\ref{fig:rho-np-graph-construction}.
  }
  \label{fig:intersectionLineGadget}
\end{figure}

\begin{lemma}\label{lem:intersectionLineGadget}
  If a graph containing the intersection line gadget as a subgraph can be embedded on two non-parallel
  planes, the vertices $v_1$, $v_2$, and $v_3$ must be drawn on the intersection
  line of the two planes while the vertices $u_1$, $u_2$, $u_3$, and $u_4$ cannot
  lie on the intersection line.
\end{lemma}
\begin{proof}
  Let $\ell$ be the intersection line of the two planes.
  For a contradiction, assume that the middle vertex $v_2$ is not drawn
  on $\ell$, but in another position on plane $P$.
  This implies that its whole neighborhood $N(v_2)$, which consists of all other
  vertices of the gadget, has to be placed on $P$, too.
  But $N(v_2) + v_2$ is a supergraph of $K_{3,3}$ and thus cannot be drawn in one
  plane.
  Hence, $v_2$ has to be drawn on $\ell$.

  Now assume that $v_1$ lies on $P$, but not on $\ell$.
  Then, again, its neighborhood -- including the vertices $u_1$, $u_2$, and
  $u_3$ -- is also drawn on $P$.
  At most two of these vertices can be drawn on $\ell$
  because they are all neighbors of $v_2$.
  Therefore, one of them, say $u_3$, lies on $P$ without $\ell$.
  As $v_3$ is a neighbor of $u_3$, it is also placed on $P$.
  Again we draw a supergraph of $K_{3,3}$ on $P$; a contradiction.

  By applying a symmetric argument to~$v_3$, we can infer that $v_1$, $v_2$,
  and $v_3$ have to be drawn on~$\ell$.
  Since $u_1, u_2, u_3, u_4$ are adjacent to all of the three vertices, they clearly
  cannot lie on~$\ell$.
\end{proof}

\tikzset{
  falseVariable/.append style={blue},
  trueVariable/.append style={red!50},
  blockingCaterpillar/.append style={thick},
}
\newcommand{\NPDrawClauseGadgets}{
  \foreach \x/\i in {0/1,2.5/2,5/3,7.5/4,10/5,12.5/6,15/7} {
    \coordinate (v1\i) at (\x-0.4,0);
    \coordinate (v2\i) at (\x,0);
    \coordinate (v3\i) at (\x+0.4,0);
    \draw (\x-0.65,0.25) rectangle (\x+0.65,-0.25);
    \draw (v1\i) -- (v2\i) -- (v3\i);
  }
}
\newcommand{\NPDrawClauseGadgetsVertices}{
  \foreach \i in {1,...,7} {
    \fill (v1\i) circle(2pt);
    \fill (v2\i) circle(2pt);
    \fill (v3\i) circle(2pt);
  }
}
\newcommand{\NPDrawBlockingCaterpillar}{
  \foreach \i in {1,2,...,7} {
    \fill (o\i) circle(2pt);
    \foreach \k in {1,2,3} {
      \draw[blockingCaterpillar] (o\i) -- (v\k\i);
    }
  }
  \draw[blockingCaterpillar] (o1) -- (o2) -- (o3) -- (o4) -- (o5) -- (o6) -- (o7);
}
\newcommand{\NPDrawVariable}[1][]{
  \begin{scope}[\col]
    \foreach \i/\x/\y/\gadgdist in \coords {
      \coordinate (\symb\i) at (\x,\y);
      \coordinate (\symb b\i) at (\x,\y+\dist);

      \ifstrequal{#1}{movef}{
        \expandafter\ifx\symb f
          \expandafter\ifx\i 2
            \coordinate (\symb b\i) at (\x,1.5+\dist);
          \fi
        \fi
      }{}
      \ifstrequal{#1}{intersection}{
        \coordinate (\symb c1\i) at (\x+\gadgdist-0.2,\y);
        \coordinate (\symb c2\i) at (\x+\gadgdist,\y);
        \coordinate (\symb c3\i) at (\x+\gadgdist+0.2,\y);
        \draw (\symb c2\i) ellipse (0.4 and 0.2);
        \foreach \k in {1,2,3} {
          \draw (\symb c\k\i) -- (\symb b\i);
          \fill (\symb c\k\i) circle(1.5pt);
        }
        \draw  (\symb c1\i) -- (\symb c3\i);
      }{}
    }
    \foreach \i/\y in \connections {
      \fill (\symb\y) circle(2pt);
      \fill (\symb b\y) circle(2pt);
      \draw (\symb\y) -- (\symb b\y);
      \foreach \k in {1,2,3} {
        \draw (\symb\y) -- (v\k\i);
      }
    }

    \draw (\symb b1) \foreach \i/\x/\y in \coords { -- (\symb b\i)};
  \end{scope}
}
\newcommand{\NPIntersectionLine}{
  \draw[dashed] (-1,0) -- (16,0);
}

\begin{theorem}\label{thm:rhoeq2}
  Given any graph~$G$, the decision problem~\enquote{$\rho_3^2(G)=2$?}
  is \cNP-hard.
\end{theorem}

\begin{proof}
  We show \cNP-hardness by reduction from \PPCOITSAT which is NP-complete by Lemma~\ref{lem:PosPlaCycOneNPC}.
  Consider an instance of that problem consisting
  of a formula~$\Phi$
  and a planar combinatorial embedding of $G(\Phi) +  C$
  where  $C=(c_1, \dots, c_n)$ is a cycle through the clause-vertices.
  This embedding gives us a partition of the vertices that are not part of~$C$:
  $V_1$ is the set of vertices in the interior of~$C$, and
  $V_2$ is the set of vertices in the exterior of~$C$.

  Let $G'$ be the graph that is obtained from $G(\Phi)$ by adding the
  edges of~$C$, except for the edge~$\{c_n,c_1\}$.
  Consider a straight-line drawing~$\Gamma$ of~$G'$ such
  that~$c_1,\dots,c_n$ lie on a horizontal line that separates the vertices
  of~$V_1$ from those of~$V_2$.  We can obtain such a drawing, for
  example, by first computing straight-line drawings of the two
  subgraphs~$G'[V_1 \cup \{c_1,\dots,c_n\}]$ and~$G'[V_2 \cup
  \{c_1,\dots,c_n\}]$ such that the path~$(c_1,\dots,c_n)$ is on the
  outer face in both drawings, and by then gluing the two drawings
  together using two appropriate linear transformations.

  We build the graph $G^*(\Phi)=(V,E)$ as follows:
  Each clause $c$ is represented by a clause gadget that consists of three
  vertices $v^1_c$, $v^2_c$, and $v^3_c$
  that are connected by a path.
  Let $x$ be a variable that occurs in the clauses $c_{i_1}, c_{i_2}, \dots,
  c_{i_l}$ with $i_1 < i_2 < \dots < i_l$.
  For the variable-clause incidence of the variable $x$ with the clause $c_k$
  ($k \in {i_1, \dots, i_l}$), we add two vertices $w^x_k$ and $v^x_k$ that are
  adjacent to each other.
  Additionally, we make~$w^x_k$ adjacent to the clause-vertices $v^1_{c_k}$,
  $v^2_{c_k}$, and $v^3_{c_k}$.
  Thus, $v^x_k$ is connected to the three clause-vertices via $w^x_k$.
  Finally, we connect the vertices $v^x_1, \dots, v^x_l$ such that they form a
  path that inherits the order~$c_1,\dots,c_n$ of the clauses on the cycle $C$.
  To each of the vertices $v^x_1, \dots, v^x_l$ one instance of the
  intersection line gadget is connected.
  Finally, we add a blocking caterpillar that consists of the vertices
  $v^\mathrm{b}_1, \dots, v^\mathrm{b}_n$ together with the clause-vertices and connects the clauses in the
  cyclic order of $C$.
  See Fig.~\ref{fig:rho-np-graph-construction} for an example of this
  construction.

\begin{figure}[tb]
  \centering
  \begin{tikzpicture}[scale=\fignphardnessscale]
    \fill[black!20] {[rounded corners=0.5cm] (-1.25,3.75) -- (-1.25,2) }
      -- (3.1,2) .. controls (3.4,2) and (3.3,2.75) .. (3.6,2.75)
      -- (13.2,2.75) .. controls (13.5,2.75) and (13.5,2) .. (13.8,2)
      {[rounded corners=0.5cm] -| (16,3.75) -- cycle};

    \NPDrawClauseGadgets

    \coordinate (o1) at (0,1);
    \coordinate (o2) at (2.5,1);
    \coordinate (o3) at (5,1);
    \coordinate (o4) at (7.5,1);
    \coordinate (o5) at (10,1);
    \coordinate (o6) at (12.5,1);
    \coordinate (o7) at (15,1);
    \NPDrawBlockingCaterpillar

    \foreach \symb/\col/\coords/\dist/\connections in
      {a/trueVariable/{1/1/-2/-1,2/6/-2/-1,3/15/-2/0.75}/-0.5/{1/1,3/2,7/3},
      b/falseVariable/{1/0/-1/-0.75,2/2.5/-1/-0.75,3/5/-1/-0.75}/-0.5/{1/1,2/2,3/3},
      c/trueVariable/{1/3.5/2/0.75,2/13.25/2/-0.75}/0.5/{2/1,6/2},
      d/falseVariable/{1/8.5/1.5/-0.55,2/9/1.5/0.75,3/11.5/1.5/-0.75}/0.5/{4/1,5/2,6/3},
      e/trueVariable/{1/8/-1.5/-0.75,2/9.5/-1.5/-0.75}/-0.5/{4/1,5/2},
      f/falseVariable/{1/5.75/1.5/-0.75,2/6.5/1.5/0.55}/0.5/{3/1,4/2},
      g/falseVariable/{1/10.5/-1.5/0.55,2/12.5/-1.5/-0.55,3/14/-1.5/-0.6}/-0.5/{5/1,6/2,7/3},
      h/falseVariable/{1/1/2.5/-1.5,2/1.5/2.5/1.4,3/13.75/2.5/1.5}/0.5/{1/1,2/2,7/3}} {

      \NPDrawVariable[intersection]
    }

    \NPDrawClauseGadgetsVertices

    \begin{scope}[black]
      \node[anchor=south] at ($ (hb1) + (-0.1,0) $) {$v^x_1$};
      \node[anchor=south] at ($ (hb2) + (0.1,0) $) {$v^x_2$};
      \node[anchor=south] at (hb3) {$v^x_3$};
      \node[anchor=east] at (h1) {$w^x_1$};
      \node[anchor=west] at ($ (h2) + (0,-0.15) $) {$w^x_2$};
      \node[anchor=west] at ($ (h3) + (0,-0.1) $) {$w^x_3$};
      \node[anchor=south] at ($ ($(hb2)!0.5!(hb3)$) + (0,0) $) {$x$};
    \end{scope}
    \node[anchor=south] at (o1) {$v^{\mathrm{b}}_1$};
    \node[anchor=south] at (o2) {$v^{\mathrm{b}}_2$};
    \node[anchor=south] at (o7) {$v^{\mathrm{b}}_7$};
    \node[anchor=east] at ($ (v11) + (-0.09,0) $) {$v^1_{c_1}$};
    \node[anchor=west] at ($ (v31) + (+0.13,0) $) {$v^3_{c_1}$};
    \node[anchor=east] at ($ (v12) + (-0.09,0) $) {$v^1_{c_2}$};
    \node[anchor=west] at ($ (v32) + (+0.13,0) $) {$v^3_{c_2}$};
    \node[anchor=east] at ($ (v17) + (-0.09,0) $) {$v^1_{c_7}$};
    \node[anchor=west] at ($ (v37) + (+0.13,0) $) {$v^3_{c_7}$};

  \end{tikzpicture}
  \caption{Example for the graph $G^*(\Phi)$ constructed from a \PPCOITSAT
    instance $\Phi$.
    The clauses are depicted by the black boxes with three vertices inside and
    denoted by $c_1, \dots, c_7$ from left to right.
    The variables are drawn in pale red (true) and blue (false).
    The variable $x$ is highlighted by a shaded background.
    The ellipses attached to variable-vertices stand for the
    intersection line gadget (see Fig.~\ref{fig:intersectionLineGadget}).
    The depicted vertices incident to the gadget correspond to~$u_1$
    in Fig.~\ref{fig:intersectionLineGadget}; $u_2$, $u_3$, and $u_4$ are not
    shown.  If $\Phi$ is true, one plane covers the blue variable
    gadgets and one plane covers the blocking caterpillar (bold black) and
    the pale red variable gadgets.}
  \label{fig:rho-np-graph-construction}
\end{figure}

  We are going to show that the formula $\Phi$ has a truth assignment with
  exactly one true variable in each clause if and only if the graph $G^*(\Phi)$
  can be drawn onto two planes.

  First, assume that the formula $\Phi$ has a \OITSAT assignment, that is, an
  assignment where exactly one variable in each clause is true.
  Then we can draw it onto two intersecting planes $\PT$,
  $\PF$ in the following way:
  We place the clause-vertices on the intersection line of the two planes in the
  order given by the cycle $C$.
  This intersection line splits each of the two planes into two half-planes.

  On the plane $\PT$ we place the variable-vertices that are set to
  true and the edges connecting them to the clause-vertices; see
  Fig.~\ref{fig:rho-np-true}.
  Obviously, we can draw the edges without crossings, because our \OITSAT
  instance is planar and each clause is connected to only one variable.
  We remark that the variables are possibly placed on both half-planes of
  $\PT$, but we use an embedding where each variable is fixed to only
  one half-plane.
  Since the true variables cover only one side of each clause gadget, we can
  attach the corresponding vertex of the blocking caterpillar to the other side.
  Note that the path connecting the vertices $v_i^\mathrm{b}$ can cross the
  intersection line between the clause gadgets because all the edges incident to
  variable-vertices of one variable stay inside one of the half-planes.

\begin{figure}[tb]
  \begin{subfigure}[b]{\textwidth}
    \centering
    \begin{tikzpicture}[scale=\fignphardnessscale]
      \NPIntersectionLine
      \NPDrawClauseGadgets
      \NPDrawClauseGadgetsVertices

      \coordinate (o1) at (0,1);
      \coordinate (o2) at (2.5,-1);
      \coordinate (o3) at (5,1);
      \coordinate (o4) at (7.5,1);
      \coordinate (o5) at (10,1);
      \coordinate (o6) at (12.5,-1);
      \coordinate (o7) at (15,1);
      \NPDrawBlockingCaterpillar

      \foreach \symb/\col/\coords/\dist/\connections in
        {a/trueVariable/{1/1/-2,2/6/-2,3/15/-2}/-0.5/{1/1,3/2,7/3},
        c/trueVariable/{1/3.5/2,2/13.25/2}/0.5/{2/1,6/2},
        e/trueVariable/{1/8/-1.5,2/9.5/-1.5}/-0.5/{4/1,5/2}} {

        \NPDrawVariable[]
      }

      \node[anchor=south] at (o1) {$v^{\mathrm{b}}_1$};
      \node[anchor=north] at (o2) {$v^{\mathrm{b}}_2$};
      \node[anchor=south] at (o3) {$v^{\mathrm{b}}_3$};
      \node[anchor=south] at (o4) {$v^{\mathrm{b}}_4$};
      \node[anchor=south] at (o5) {$v^{\mathrm{b}}_5$};
      \node[anchor=north] at (o6) {$v^{\mathrm{b}}_6$};
      \node[anchor=south] at (o7) {$v^{\mathrm{b}}_7$};
    \end{tikzpicture}
    \caption{Plane $P_\mathrm{T}$ for the instance shown in
      Fig.~\ref{fig:rho-np-graph-construction}.
      The true variables are drawn in red.
      The remaining black vertices and edges form the blocking caterpillar.
      }
    \label{fig:rho-np-true}
  \end{subfigure}

  \medskip

  \begin{subfigure}[b]{\textwidth}
    \centering
    \begin{tikzpicture}[scale=\fignphardnessscale]
      \NPIntersectionLine
      \NPDrawClauseGadgets
      \NPDrawClauseGadgetsVertices

      \foreach \symb/\col/\coords/\dist/\connections in
        {b/falseVariable/{1/0/-1,2/2.5/-1,3/5/-1}/-0.5/{1/1,2/2,3/3},
        d/falseVariable/{1/8.5/1.5,2/9/1.5,3/11.5/1.5}/0.5/{4/1,5/2,6/3},
        f/falseVariable/{1/5/1.5,2/6.75/-1}/0.5/{3/1,4/2},
        g/falseVariable/{1/10.5/-1.5,2/12.5/-1.5,3/14/-1.5}/-0.5/{5/1,6/2,7/3},
        h/falseVariable/{1/1/2.5,2/1.5/2.5,3/13.75/2.5}/0.5/{1/1,2/2,7/3}} {

        \NPDrawVariable[movef]
      }
      \node[anchor=south] at ($ (fb1) + (-0.1,0) $) {$v^y_1$};
      \node[anchor=south] at ($ (fb2) + (0.1,0) $) {$v^y_2$};
      \node[anchor=east] at (f1) {$w^y_1$};
      \node[anchor=north] at (f2) {$w^y_2$};
    \end{tikzpicture}
    \caption{Plane $P_\mathrm{F}$ for the instance shown in
      Fig.~\ref{fig:rho-np-graph-construction}.
      The false variables are drawn in blue.
      At the clause gadget in the middle variable $y$ crosses
      the intersection line of the two planes.
      }
    \label{fig:rho-np-false}
  \end{subfigure}

  \caption{The two planes for the instance shown in
    Fig.~\ref{fig:rho-np-graph-construction}.
    The clauses are depicted by the black boxes with three vertices inside.
    The dashed line is the intersection line between the two planes.
    At the clause gadget in the middle of the line one of the variables crosses
    the intersection line of the two planes.
    The intersection line gadgets are left out in these figures, but can easily
    placed on the dashed line without interfering.}
\end{figure}

  On the plane $\PF$ we place the false variables' vertices; see
  Fig.~\ref{fig:rho-np-false}.
  Each clause contains exactly two false variables, which we obviously have to
  place on different half-planes.
  If both variables are on the same side of the cycle $C$ in the given planar
  embedding $\Gamma$, we draw one of the vertices $w^x_i$ onto the other side.
  Since we have only two variables per clause, we can draw the edge $\{v^x_i,
  w^x_i\}$, which connects the two half-planes, directly alongside the clause
  gadget without destroying planarity.
  Clearly, we can add the remaining vertices of the intersection line gadgets
  without interfering each other.

  For the other direction, we assume that we are given a drawing of $G^*(\Phi)$
  on two planes.
  Obviously, these planes intersect as otherwise the intersection line gadgets
  contained in $G^*(\Phi)$ could not be drawn.
  Lemma~\ref{lem:intersectionLineGadget} shows that the clause-vertices lie on
  the intersection line of the two planes, while the variable-vertices and the
  blocking caterpillar cannot lie on the intersection line.
  The vertices of each variable completely lie on one plane:
  Since they are connected, one of them had to be placed on the intersection
  line otherwise to prevent edges running outside of the planes.
  Similarly the blocking caterpillar is only on one of the planes; we call this
  plane $\PT$, the other one $\PF$.

  To get a \OITSAT assignment for $\Phi$, we now set the variables that are
  drawn on $\PT$ to true and those on $\PF$ to false.
  Obviously, every clause gadget can have at most one neighbor in each of the
  four half-planes.
  Since each clause is adjacent to a vertex of the blocking caterpillar in one
  of the half-planes of
  $\PT$, it is connected to at most one variable in $\PT$; that is, each clause
  contains at most one true variable.
  On the other hand, only two variables of the clause can be drawn on $\PF$, so
  there are at most two false variables in each clause.
  Together this yields that there is exactly one true variable in each clause
  and thus we constructed a feasible \OITSAT assignment for $\Phi$.
\end{proof}

\begin{corollary} \label{cor:rhogeq2}
  Given any graph~$G$ and any integer~$k \geq 2$, the decision
  problem~\enquote{$\rho_3^2(G) = k$?} is \cNP-hard.
\end{corollary}
\begin{proof}
  We extend the approach from Theorem~\ref{thm:rhoeq2} by additional blocking
  gadgets.
  We add $(k-2)$ identical gadgets, each of which blocks one plane.
  The gadget consists of two new vertices $u$ and $v$, which are adjacent.
  Each of $u$ and $v$ is adjacent to all vertices in every clause gadget.
  A drawing of the gadget is depicted in Fig.~\ref{fig:rho-np-coro}.

  For a given variable assignment, we can easily find a drawing:
  Use $k$ different planes that share one common intersection line.
  Obviously, we can place each of the new gadgets onto one plane with no
  additional vertices on them.
  The variable-vertices are on the remaining two planes as described in
  Theorem~\ref{thm:rhoeq2}.

  For the other direction, we first discuss the arrangement of the $k$ planes.
  They also have to be placed in a way that they share one common intersection
  line because otherwise the clause gadgets, which are part of a
  $K_{3,2k}$ subgraph, could not be drawn.
  Each of the new blocking gadgets has to use an individual plane.
  The variable-vertices are forced on the ``true'' and ``false'' plane as in
  Theorem~\ref{thm:rhoeq2}.
\end{proof}
\begin{figure}
  \centering
  \begin{tikzpicture}[scale=\fignphardnessscale]
    \NPIntersectionLine
    \NPDrawClauseGadgets
    \NPDrawClauseGadgetsVertices

    \coordinate (o1) at (8.75,2);
    \coordinate (o2) at (8.75,-2);
    \fill (o1) circle(2pt);
    \fill (o2) circle(2pt);
    \foreach \i in {1,2,...,7} {
      \foreach \k in {1,2,3} {
        \draw (o1) -- (v\k\i);
        \draw (o2) -- (v\k\i);
      }
    }
    \draw (o1) -- (o2);
    \node[above] at (o1) {$u$};
    \node[below] at (o2) {$v$};
  \end{tikzpicture}
  \caption{
    Blocking gadget that occupies a whole plane.
    Again, the clauses are depicted by the boxes and the dashed line is the
    intersection of the $k$ planes.
  }
  \label{fig:rho-np-coro}
\end{figure}

\section{Complexity of the Weak Affine Cover Numbers~
  $\pi^1_3$ and~$\pi^2_3$}
\label{s:weak-appendix}

Recall that a \emph{linear forest} is a forest whose connected components
are paths. The \emph{linear vertex arboricity} $\lva G$ of a graph~$G$
equals the smallest size $r$ of a partition $V(G)=V_1\cup\dots\cup V_r$
such that the induced subgraphs $G[V_1],\dots,G[V_r]$ are all linear forests.
The \emph{vertex thickness} $\vt(G)$ of a graph $G$ is the smallest
size~$r$ of a partition $V(G)=V_1\cup\dots\cup V_r$ such that the
induced subgraphs $G[V_1],\dots,G[V_r]$ are all planar.
Obviously, $\vt(G)\le\lva G$.
We recently used these notions to characterize the 3D weak affine cover
numbers in purely combinatorial terms~\cite{ChaplickFLRVW20}:
$\pi^1_3(G)=\lva G$ and $\pi^2_3(G)=\vt(G)$.

\long\gdef\contentThmComplPiOneThree{%
  For $l\in\{1,2\}$, given any graph~$G$,
    \begin{enumerate}
        \item\label{thm:ComplPi13:NPcomplete}
        deciding whether or not $\pi^l_3(G)\le2$ is \cNP-complete, and
        \item\label{thm:ComplPi13:approximation}
        approximating $\pi^l_3(G)$ within a factor of $O(n^{1-\epsilon})$,
        for any $\epsilon>0$, is \cNP-hard.
    \end{enumerate}
}

\begin{theorem}
  \label{thm:ComplPi13}
  \contentThmComplPiOneThree
\end{theorem}

\begin{proof}
        {\color{lipicsGray}\sffamily\bfseries\ref{thm:ComplPi13:NPcomplete}.}
        The membership in \cNP follows directly from the above combinatorial
        characteriza\-tion~\cite{ChaplickFLRVW20},
        which also allows us to deduce \cNP-hardness
        from a much more general hardness result by
        Farrugia~\cite{f-vpfai-EJC04}: For any two graph classes $\mathcal
        P$ and $\mathcal Q$ that are closed under vertex-disjoint unions
        and taking induced subgraphs, deciding whether the vertex set of a
        given graph $G$ can be partitioned into two parts $X$ and $Y$ such
        that $G[X]\in\mathcal P$ and $G[Y]\in\mathcal Q$ is \cNP-hard unless
        both $\mathcal P$ and $\mathcal Q$ consist of all graphs or all
        empty graphs.
        To see the hardness of our two problems,
        we set $\mathcal{P} = \mathcal{Q}$ to the class of linear forests
        (for $l=1$) and to the class of planar graphs (for $l=2$).

        {\textcolor{lipicsGray}{\sffamily\bfseries\ref{thm:ComplPi13:approximation}.}}
        The combinatorial characterization~\cite{ChaplickFLRVW20}
        given above
        implies that $\chi(G) \le 4 \vt(G) = 4 \pi^2_3(G)$ (by the
        four-color theorem).  Note that each color class can be placed
        on its own line, so $\pi^1_3(G)\le \chi(G)$.
        As $\pi^2_3(G) \leq \pi^1_3(G)$, both parameters are linearly
        related to the chromatic number of $G$.
        Now, the approximation hardness of our problems follows from that of the
        chromatic number~\cite{z-ldeim-TC07}.
\end{proof}

\section{Conclusion and Open Problems}
\label{s:open}

We have determined the complexity of computing the line cover numbers
$\rho^1_2$ and $\rho^1_3$, which turned out to be $\erclass$-complete.
On the positive side, these problems admit an FPT algorithm
(Corollary~\ref{cor:line-cover-fpt}).  This is impossible for the
plane cover number $\rho^2_3$ (unless $\cP=\cNP$) because the decision
problem \enquote{$\rho^2_3(G) \le k$?} is \cNP-hard even for $k=2$
  (Theorem~\ref{thm:rhoeq2}).
  If $k$ is arbitrary and given as a part of the input, then this
  problem is in~$\erclass$ (Lemma~\ref{lem:inER})---but is it
  $\erclass$-hard?

  For a graph~$G$, its affine cover number $\rho^1_2(G)$ is connected
  to its segment number $\segm(G)$ that we mentioned in the
  introduction.  A 3D version and other variants of the segment number
  have been studied by Okamoto et al.~\cite{orw-ysng-GD19}.  They have
  shown that it is $\erclass$-complete to compute the 2D and 3D
  segment numbers \cite{orw-ysng-GD19}~-- as the 2D and 3D line cover
  numbers.  It seems plausible that the 2D and 3D segment numbers also
  admit FPT algorithms to compute them.

  Our proof of Theorem~\ref{thm:rho12-hard} implies that computing
  $\rho^1_2(G)$ and $\rho^1_3(G)$ is hard even for planar graphs of
  maximum degree~$4$. Can $\rho^1_2(G)$ and $\rho^1_3(G)$ be computed
  efficiently for trees? %
  Note that
  this is true for the segment number $\segm(G)$~\cite{desw-dpgfs-CGTA07}.

Concerning the weak affine cover numbers, we established the
computational hardness of $\pi^1_3(G)$ and $\pi^2_3(G)$
in Theorem~\ref{thm:ComplPi13}\ref{thm:ComplPi13:NPcomplete}.
Recently, Biedl et al.~\cite{bfw-lpcnr-GD19} showed that it is
\cNP-complete to decide whether $\pi^1_2(G)=2$ for a given planar graph~$G$.
  By Theorem~\ref{thm:ComplPi13}\ref{thm:ComplPi13:approximation},
  the weak affine cover numbers
  $\pi^1_3(G)$ and $\pi^2_3(G)$ are even hard to approximate.
  How hard is it to approximate $\rho^1_2(G)$, $\rho^1_3(G)$,
  $\rho^2_3(G)$, and~$\pi^1_2(G)$ for a given (planar) graph~$G$?

\subsubsection*{Acknowledgments.}

We thank Stefan Kratsch for a useful discussion of our FPT results in
Section~\ref{s:fpt}.

\bibliographystyle{abbrvurl}
\bibliography{abbrv,planes}

\end{document}